\documentclass{ifacconf}   
\let\theoremstyle\relax
\usepackage{amsmath,amsfonts,amssymb,amsthm}
\usepackage{bm}
\usepackage{natbib}       
\usepackage{xcolor}
\usepackage{algpseudocode}
\usepackage{algorithm}
\usepackage[short]{optidef}
\usepackage{appendix}
\usepackage{tikz}
\usepackage{graphicx}
\usetikzlibrary{arrows.meta, positioning, decorations.pathmorphing}
\theoremstyle{remark}
\newtheorem{definition}{Definition}
\AtBeginEnvironment{definition}{%
  \pushQED{\qed}%
}
\AtEndEnvironment{definition}{\popQED\endexample}
\newtheorem{proposition}{Proposition}
\newtheorem{example}{Example}

\newtheorem{remark}{Remark}

\newtheorem{theorem}{Theorem}
\newtheorem{assumption}{Assumption}
\newtheorem{corollary}{Corollary}
\makeatletter
\newcommand\fs@betterruled{%
  \def\@fs@cfont{\bfseries}\let\@fs@capt\floatc@ruled
  \def\@fs@pre{\vspace*{9pt}\hrule height.4pt depth0pt \kern 1.8pt}%
  \def\@fs@post{\kern1pt\hrule\kern-3pt}%
  \def\@fs@mid{\kern0.6pt\hrule\kern -3pt}%
  \let\@fs@iftopcapt\iftrue}
\floatstyle{betterruled}
\restylefloat{algorithm}
\makeatother

\DeclareMathAlphabet{\mymathbbb}{U}{BOONDOX-ds}{m}{n}

\newcommand{\mr}[1]{\mathrm{#1}}
\newcommand{\rrr}{\mathbb{R}}

\newcommand{\zzz}{\mathbb{Z}}

\newcommand{\At}{A_{\mathrm{t}}}
\newcommand{\Bt}{B_{\mathrm{t}}}
\newcommand{\Aq}{A_{\mathrm{q}}}
\newcommand{\Bq}{B_{\mathrm{q}}}
\newcommand{\xt}{x_{\mathrm{t}}}
\newcommand{\xq}{x_{\mathrm{q}}}
\newcommand{\xte}{x_{\mathrm{te}}}
\newcommand{\xqe}{x_{\mathrm{qe}}}
\newcommand{\yte}{y_{\mathrm{te}}}
\newcommand{\yqe}{y_{\mathrm{qe}}}
\newcommand{\zte}{z_{\mathrm{te}}}
\newcommand{\zqe}{z  _{\mathrm{qe}}}

\newcommand{\eps}{\varepsilon}

\newcommand{\one}{\mathbf{1}}
\newcommand{\zero}{\mathbf{0}}

\newcommand{\Reps}{\rrr_\varepsilon}
\newcommand{\Rtop}{\rrr_\top}
\newcommand{\matr}[2]{{ \left[ \begin{array}{#1} #2 \end{array} \right] }}

\newcommand{\R}{\mathbb{R}}
\newcommand{\cR}{\mathcal{R}}

\newcommand{\Tx}{\Tilde{x}}
\newcommand{\TA}{\Tilde{A}}
\newcommand{\TB}{\Tilde{B}}
\newcommand{\pnorm}[1]{||#1||_\mathbb{P}}

\begin{document}
\begin{frontmatter}

\title{Dynamics of Implicit Time-Invariant Max-Min-Plus-Scaling Discrete-Event Systems} 

\author[First]{Sreeshma Markkassery} 
\author[First]{Ton van den Boom} 
\author[First]{Bart De Schutter}

\address[First]{Delft Center for Systems and Control (DCSC), Delft University of Technology, The Netherlands
  (e-mail: \{s.markkassery, A.J.J.vandenBoom, B.DeSchutter\}@tudelft.nl)}

\begin{abstract}  
Max-min-plus-scaling (MMPS) systems generalize max-plus, min-plus and max-min-plus models with more flexibility in modelling discrete-event dynamics. Especially, implicit MMPS models capture a wide range of real world discrete-event applications. This article analyzes the dynamics of an autonomous, time-invariant implicit MMPS system in a discrete-event framework. First, we provide sufficient conditions under which an implicit MMPS system admits at least one solution to its state-space representation. Then, we analyze its global behavior by determining the key parameters; the growth rates and fixed points. For a solvable MMPS system, we assess the local behavior of the system around its set of fixed points via a normalization procedure.  Further, we present the notion of stability for the normalized system. A case study of the urban railway network substantiates the theoretical results. 
\end{abstract}


\end{frontmatter}
\section{Introduction}
Max-plus algebraic models are widely used to capture the timing behavior of discrete-event systems (DES) \citep{cassandras1995introduction}. Over time, several extensions to these models have been developed, including min-plus algebraic systems, max-min-plus systems, switching max-plus linear systems, and other related models \citep{baccelli1992synchronization, mpw,van2006modelling}. Max-min-plus-scaling (MMPS) systems are the recent extension within this framework \citep{mmps} that employ operations such as maximization, minimization, addition, and scaling to enable more flexible modelling of discrete-event dynamics. 

In parallel, applications with discrete-event dynamics are also modelled as conventional nonlinear systems. For example, nonlinear discrete-time models for passenger-oriented control and scheduling in urban rail systems have been explored in \citep{liu2022modeling, passenger, wang2017train}. Analyzing the steady-state behavior of such models usually gives rise to non-smooth, non-convex optimization problems. Also, there are no general modelling criteria in conventional algebra that can be used for discrete-event applications. In contrast, max-plus algebra provides a mathematical framework for modelling discrete-event systems. Discrete-event systems involving only synchronization can be modelled using max and plus operations \citep{baccelli1992synchronization, mpw}. For instance, metro traffic flow can be modeled using max-plus systems, as shown in \cite{mplrailway}, though this model neglects passenger demand. Modelling competitive interactions requires a min operation that can be handled by max-min-plus systems \citep{gunawardena1994timing, olsder1994structural}. In cases like an urban railway line, passenger flow has a scalar dependency on the arrival and departure times of the train, which can be modelled using a scaling operation. By combining synchronization, competition, and scaling, MMPS systems enable more realistic representation of complex discrete-event systems. However, their nonlinearity falls outside classical max-plus theory, requiring a new mathematical framework to analyze their stationary behavior and stability.

In practice, discrete-event dynamic systems modelled using max-plus algebra (or min-plus algebra) are often implicit, i.e., the states `$x(k)$' in the same event cycle `$k$' depend on each other.  Such implicit dependencies are seen in examples like traffic networks, railway lines, manufacturing systems, etc \citep{baccelli1992synchronization, modelmatching,alirezaei2012max, komenda2018max}. 

The linear implicit max-plus algebraic models can be written into an explicit form using the Kleene star product \citep{mpw}. However, for nonlinear models in max-plus algebra (and min-plus algebra) like max-min-plus and MMPS systems, it is far more difficult to obtain an explicit form of the implicit model. Hence, the goal of this paper is to analyze the global parameters; growth rates and fixed points, and the local stability of implicit MMPS systems without converting them to an explicit form. 

Max-plus linear systems and max-min-plus systems have been studied in \citet{baccelli1992synchronization,mpw,eigmmp,gunawardena1994min,gunawardena1994timing}. Most of this literature is centered around the existence and uniqueness of the growth rate and stationary regime of max-plus algebraic systems that are governed by homogeneous, monotonic, and non-expansive functions. These functions have a unique growth rate \citep{gunawardena1994min}. MMPS systems, in general, lack the property of monotonicity and non-expansiveness because of the scaling operation \citep{eigmmps,mmpsWODES2024}. Hence, an MMPS system may not always have a unique growth rate. It can have multiple stationary regimes and growth rates based on the initial condition of the system. \citet{mmpsWODES2024,eigmmps} discuss the asymptotic behavior, calculation of growth rates, and fixed points of a homogeneous explicit MMPS system. 

Likewise, non-monotonic and expansive functions have been studied in various contexts. For example, the issue of analyzing the asymptotic behavior of a max-plus-scaling system is stated as an open problem in \citet{plus1999max}.
Homogeneous, non-expansive systems has been studied in \citet{cgq95,cgq98, abg21}. The latter considers min, plus and
scaling operations, showing that throughput (which is similar to growth rate) can be computed
in polynomial time by reducing to a Markov decision process.
Additionally, time-invariant, monotonic MMPS dynamics are equivalent to the dynamic programming equations of turn-based stochastic games \citep{agnt23}. Hence, studying MMPS systems provides insights into problems like throughput computation and mean payoff in stochastic turn-based games.

Another motivation to study implicit MMPS systems is to have a better understanding of the autonomous stability of the discrete-event system. Stability analysis of a dynamical system is important in developing suitable closed-loop control for the system. A convenient framework to examine the stability of a discrete-event system has been provided in \citet{baccelli1992synchronization, commault1998feedback}. The stability notions introduced in \citet{gupta2020framework} is one of the recent studies that assess the autonomous stability of a switching max-plus linear system. The main results of stability analysis of max-plus systems and max-min-plus systems were about the existence and calculation of fixed points \citep{olsder1994structural, cochet1999constructive, van2001characterization}. The existence of a fixed point guarantees stability of max-plus linear and max-min-plus systems. However, for MMPS systems, fixed points can exist without the system being stable. A local stability analysis of an explicit MMPS model is discussed in \citep{mmpsWODES2024}. In the current paper, we extend this analysis to incorporate implicit MMPS models as well.

The major contributions of this article are with regard to analyzing the dynamics of autonomous implicit MMPS systems that are time-invariant. First, we discuss the conditions (Theorem \ref{th:unique}) under which an implicit MMPS system is solvable. Then, for such solvable implicit MMPS systems, we derive an algorithm to calculate the fixed points and growth rates via a normalization procedure. We show that there may exist multiple fixed points corresponding to a single growth rate, and we provide a procedure to calculate this fixed point set. For a specific growth rate and fixed point set, we derive a normalized implicit MMPS system that can be mapped to a linear dynamical system with polyhedral constraints (Proposition \ref{prop:linsys}). Next, we assess the local max-plus bounded buffer stability of this mapped linear system by using tools from linear discrete-time stability analysis. Finally, we present a case study of analysis of an urban railway network modeled as an implicit MMPS system.    

The rest of this paper is organized as follows. In Section \ref{sec2}, we provide the background on max-plus algebra and MMPS systems. In Section \ref{sec_impsys}, we introduce a canonical form for an implicit autonomous MMPS system and present its important properties. Section \ref{solvabilityimp} deals with the solvability of an implicit autonomous MMPS system. In Section \ref{sectiongrowthrate}, we propose a novel methodology to derive a normalized system from an implicit MMPS system. Then we present the local stability of an implicit autonomous MMPS system in Section \ref{sec4} using the normalized system. Section \ref{sec:urban} presents a case study of an urban railway network modeled as an implicit MMPS system. We discuss the conclusions in Section \ref{sec7}. 
\section{Mathematical Preliminaries}
\label{sec2}
Let $\top=\infty$, $\varepsilon=-\infty$, $\rrr_\top = \rrr \cup \{\infty\}$, $\rrr_\varepsilon = \rrr \cup \{-\infty\}$, and let $\zzz^+$, denote the set of positive integers \citep{mpw}. The set $\cR$ indicates one of the following sets: $\rrr$, $\Reps$, or $\Rtop$. The notations $\mathbf{1}$ and $\mathbf{0}$ are used to denote the column vector with all components equal to $1$ and the zero vector of the appropriate dimension, respectively. In some cases, the notations $\one_n$ and $\zero_n$ specify the dimension, $n$, of the vectors $\one$ and $\zero$, respectively. The matrix $I_n$ is an identity matrix of size $n\times n$ in the conventional algebra. Let $\overline{n}$ denote the set of all positive integers up to $n\in \zzz^+$. The notation `$\intercal$' is used to denote transpose of a matrix/vector. For $a,b \in \cR$, the operations $a\oplus b = \max{(a,b)}$ and $a\otimes b = a+b$ are called max-plus addition and max-plus multiplication. The operations $a\oplus'b = \min{(a,b)}$ and $a\otimes'b = a+b$ are called min-plus addition and min-plus multiplication. Max-plus addition, max-plus multiplication, min-plus addition and min-plus multiplication \citep{mpw} of matrices $A, B \in \cR^{m \times n}$  and $C \in \cR^{n \times p}$ are defined as follows: 
\begin{align*}
\label{max-plusaddmult}
[ A\! \oplus\! B ]_{ij} &\!=\! \max([A]_{ij},[B]_{ij}\!), 
[ A\! \otimes\! C ]_{ij} \!=\! \max_{k\in \overline{n}}  ([A]_{ik}\! +\! [C]_{kj}\!) \\
[ A\! \oplus'\! B ]_{ij}  &\!=\! \min([A]_{ij},[B]_{ij}\!), 
[ A \!\otimes'\! C ]_{ij} \!=\! \min_{k\in\overline{n}}  ([A]_{ik} \!+\! [C]_{kj}\!).
\end{align*}
From conventional algebra and matrix theory, we use the following notations and definitions. The matrix-vector product of the matrix, $C\in \rrr^{n\times p}$ and a vector, $x\in \rrr^p$ is denoted as $C\cdot x$ and the scalar multiplication of a scalar $\mu \in \rrr$ and the vector $x$ is denoted as $\mu x$. The standard basis vector is denoted as a row vector of appropriate size, $\bm{e}_j$, with the $j$-th component equal to $1$ and other components equal to 0. 
\begin{definition}[Kronecker Product]
\label{kronecker}
    The Kronecker product of a matrix $A$ and vector $\one_n$, $A\boxtimes\one_n$ stacks $n$ copies of every row of the matrix $A$ vertically and $\one_n\boxtimes A$ stacks $n$ copies of the entire $A$ matrix vertically.
\end{definition}
\begin{definition}[Row-major stacking of a matrix]
\label{rowmajor}
    The row-major stacking of a matrix $A\in \cR^{n\times m}$ is the stacking of rows of a matrix into a column vector, vec$(A)$ as follows:
    $$\text{vec}(A) = [A_{1,\cdot}^{\intercal}~~A_{2,\cdot}^{\intercal}\hdots A_{n,\cdot}^{\intercal}]^{\intercal}$$
    where $A_{i,\cdot}, i\in \overline{n}$ denotes the $i$-th row of $A$.
\end{definition}
\begin{definition}[Permutation matrix]
    A square matrix $T\in \R^{n\times n}$ is a permutation matrix if there is exactly one $1$ in each row and column of the matrix and all the other entries are $0$. Multiplication by such matrices results in the permutation of the rows or columns of the other matrix.
\end{definition}
\begin{definition}\label{diagmat}
    Given the vector $v \in \rrr^n$, we define a max-plus diagonal matrix, $\text{d}_\otimes(v)$  and the min-plus diagonal matrix, $\text{d}_{\otimes'}(v)$
\begin{align*}
    \text{d}_\otimes(v) 
  = \matr{cccc}{v_1    & \eps   & \cdots & \eps   \\
                \eps   &  v_2   & 		 & \vdots \\
                \vdots &        & \ddots & \vdots \\
                \eps   & \cdots & \cdots & v_n    },
    \text{d}_{\otimes'}(v) 
  = \matr{cccc}{v_1    & \top   & \cdots & \top   \\
    	        \top   &  v_2   & 		 & \vdots \\
	            \vdots &        & \ddots & \vdots \\
	            \top   & \cdots & \cdots & v_n    }.
\end{align*}
The inverse max-plus diagonal matrix is  $ [\text{d}_\otimes(v) ]^{-1}=\text{d}_{\otimes}(-v)$ and the inverse min-plus diagonal matrix is $ [\text{d}_{\otimes'}(v)]^{-1} =\text{d}_{\otimes'}(-v)$.
\end{definition}
\begin{definition}[\cite{mpw}]
    A matrix $A\in \mathcal{R}^{n\times m}$ is said to be regular if $A$ has at least one finite element in each row.  
\end{definition}
\begin{definition}[\cite{mmps}]
	A max-min-plus-scaling (MMPS) function $f:\cR^m\to\cR$ of the variables $x_1,\dots,x_m \in \cR$ is defined by the grammar
	\begin{equation*}
		f:=x_i|\alpha|\max(f_k,f_l)|\min(f_k, f_l)| f_k + f_l|\beta f_k,\;  
	\end{equation*}
	where $\alpha \in \cR,\;\beta\in \rrr$ are some scalars and $f_k$, $f_l$ are MMPS functions. 
	For vector-valued MMPS functions, the above statements hold component wise. This definition is the `Backus-Naur' form from computer science, where the vertical bars separate the different ways by which the function can be recursively constructed.
 \end{definition}
\begin{definition}[\cite{mmps}]\label{def:MMPSsys}Consider the vector 
	\[ \chi(k) = \left[x^\intercal(k),\,x^\intercal(k-1),\ldots,  u^\intercal(k),\,w^\intercal(k)\right]^\intercal,\] 
	where $\chi \in \mathcal{X}\subseteq \cR^{n_\mathrm{p}}$, $x \in \cR^n$ is the state, $u \in \cR^q$ is the control input, and $w \in \cR^z$ is an external signal.  
	A max-min-plus-scaling (MMPS) system describes a state space model of the form 
	\begin{equation}\label{eq:MMPS}
		\begin{aligned}
			x(k) &= f(\chi(k)),
		\end{aligned}
	\end{equation}
	where $f$ is a vector-valued MMPS function in the variable $\chi$.
\end{definition} 
For implicit autonomous MMPS systems, we obtain $\chi(k) = [x^\intercal(k-1),\,x^\intercal(k)]^\intercal$, while for explicit autonomous MMPS systems we obtain $\chi(k) = x(k-1).$
\begin{definition}[Solvability of MMPS system]
    The MMPS system \eqref{eq:MMPS} is considered solvable if, for every $x(k-1)\in \rrr^n$, there exists at least one state $x(k)$ such that the dynamics of the state space description hold.
\end{definition}
\begin{definition}[\cite{eigmmps}]\label{def:ABC_MMPS}(ABC canonical form)
The autonomous explicit MMPS system $x(k) = f(x(k-1))$ with $f$ being an MMPS function can be reformulated as the following disjunctive canonical form:
 \begin{equation}\label{abcmmps}
	    x(k) = A \otimes ( B \otimes' (C\cdot x(k-1)))
	\end{equation}
	for some matrices $A\in\rrr_\varepsilon^{n\times m}$, $B\in \rrr_{\top}^{m\times p}$, $C\in\rrr^{p\times n}$, and $x\in\R^{n}, k\in \zzz^{+}$.
\end{definition} 
Note that the explicit MMPS systems \eqref{abcmmps} with regular $A$ and $B$ matrices are always solvable.  

In max-plus linear systems and max-min-plus systems, we typically focus on states having the dimension of time (temporal states $\xt$) that grow with a bounded growth rate. However, discrete-event systems can also have quantity states $\xq$ (e.g. the number of passengers on a train in an urban railway network) that will not grow unbounded, but that will remain bounded. Having two different types of states helps in the efficient modelling of discrete-event systems \citep{mmps}.
\begin{definition}[\cite{mmps}]\label{partial additive homogenety}
    (Partial additive homogeneity) Consider $\xt\in\R^{n_\mathrm{t}}$ and $\xq\in\R^{n_\mathrm{q}}$ and MMPS functions $f_\mathrm{t}:\R^{n_\mathrm{t}\times n_\mathrm{q}} \rightarrow \R^{n_\mathrm{t}}$ and $f_\mr{q}: \R^{n_\mathrm{t}\times n_\mathrm{q}} \rightarrow \R^{n_\mathrm{q}}$. Then, the system 
    \begin{align*}
        \begin{bmatrix}
        \xt(k)\\\xq(k)\end{bmatrix} = \begin{bmatrix}f_\mr{t}(\chi_\mr{t}(k),\chi_\mr{q}(k))\\f_\mr{q}(\chi_\mr{t}(k),\chi_\mr{q}(k))\end{bmatrix}
    \end{align*} where $\chi_\mr{t}(k) = [\xt^\intercal(k),\,\xt^\intercal(k-1)]^\intercal$, $\chi_\mr{q}(k) = [\xq^\intercal(k),\,\xq^\intercal(k-1)]^\intercal$, is partly additively homogeneous if for any $h\in \R$ the following holds:
    \begin{align}
    \begin{aligned}
        f_\mr{t}(\chi_\mr{t}+h\one,\chi_\mr{q}) &= f_\mr{t}(\chi_\mr{t},\chi_\mr{q})+h\one,\\
        f_\mr{q}(\chi_\mr{t}+h\one,\chi_\mr{q}) &= f_\mr{q}(\chi_\mr{t},\chi_\mr{q}).
         \end{aligned}
         \label{ti:prop}
    \end{align}
    The subscript `$\mr{t}$' is associated with temporal states and `$\mr{q}$' is associated with quantity states.
\end{definition}
An autonomous MMPS system is time-invariant if it is partly additively homogeneous.  
\begin{definition}
  (\citep{mmpsWODES2024} Growth rate, fixed point) The time-invariant MMPS system $x(k) = f(\chi(k)))$, $x \in \cR^n$ and $f:\cR^n \rightarrow \cR^n$ is said to have a growth rate (additive eigenvalue) if there exists a real number $\lambda \in \R$ and a vector $v \in \rrr^n$ such that
$$f(v) = v + \lambda s,$$
where $s = [\one_{n_\mr{t}}^{\intercal}~\zero_{n_\mr{q}}^{\intercal}]^{\intercal}$, $n_\mr{t}$ and $n_\mr{q}$ are the number of temporal states and quantity states, respectively. The scalar $\lambda$ is then called a growth rate and the vector $v$ is called a corresponding fixed point (additive eigenvector). Further, if $v$ is a fixed point, $v+h[\one_{n_\mr{t}}^{\intercal}~\zero_{n_\mr{q}}^{\intercal}]^{\intercal}$ is also a fixed point for any $h\in \rrr$.
\end{definition} 

\begin{definition}
  (Multiplicative eigenvalue, multiplicative eigenvector) The linear system $x(k) = M\cdot x(k-1)$, $x \in \R^n$ and $M\in\R^{n\times n}$ is said to have a multiplicative eigenvalue if there exists a real number $\mu \in \R$ and a vector $w\neq \zero \in \R^n$ such that
$M \cdot w = \mu  w$.
The scalar $\mu$ is then called a multiplicative eigenvalue and $w$ is a corresponding multiplicative eigenvector. 
\end{definition} 

\begin{definition}[\cite{mpw}]
    The Hilbert projective norm of a vector $x\in\rrr^n$ in max-plus algebra is defined as 
    \begin{align*}
    \pnorm{x} = \underset{i\in \overline{n}}{\max} \;x_i - \underset{i\in \overline{n}}{\min}\;x_j.
    \end{align*}
\end{definition}
\begin{definition}[\cite{gupta2020framework}]
    An autonomous discrete-event system is max-plus bounded buffer stable if for every initial state, $x_0\in \rrr^n$, there exists a bound $m(x_0)\in \rrr$ such that the states are bounded in Hilbert's projective norm: $\pnorm{x(k)} \leq m(x_0) \;\forall k\in \zzz^+$. 
\end{definition}
This implies that all the temporal states, i.e., the states with the dimension of time, have same growth rate and all the quantity states have zero growth rate.
\section{Implicit Autonomous MMPS systems}
\label{sec_impsys}
Similar to the ABC-canonical form for an explicit MMPS system introduced in \citep{eigmmps}, we define the following form for an implicit MMPS system.
\begin{definition}(Implicit ABCD canonical form)
Consider the following system:
	\begin{equation} 
		  x(k) = A \otimes ( B \otimes' (C\cdot x(k-1) + D\cdot x(k) )).
		  \label{impsys}
	\end{equation}	  
	This system is an implicit MMPS system in the disjunctive ABCD canonical form.
\end{definition} 
It can be easily observed that the implicit ABCD canonical form is a superset of the explicit canonical form \eqref{abcmmps}. When $D=0$ in \eqref{impsys}, we get the explicit form \eqref{abcmmps}. Thereupon, all the analysis methods presented next are more general and applicable to explicit systems as well. 
\begin{proposition}
Any implicit MMPS system (\ref{eq:MMPS}) can be written in the ABCD canonical form.
\end{proposition} 
\begin{proof}
    In \citep{eigmmps}, we proved that any explicit MMPS system can be written in the ABC canonical form. The proof for implicit MMPS system can be formulated in the same way.
\end{proof}
Note that this canonical form is not unique. In addition, any implicit MMPS system can be written in a conjunctive canonical form as 
\begin{align}
\label{second representation}
    x(k) = A_2\otimes'(B_2\otimes (C_2\cdot x(k-1)+D_2\cdot x(k))),
\end{align}
where $A_2\in \rrr_{\top}^{n\times \tilde{m}}$, $B_2\in \rrr_{\varepsilon}^{\tilde{m}\times p}$ and $C_2, D_2\in \rrr^{p\times n}$. The order of $\otimes'$ and $\otimes$ operations does not matter in representing an MMPS system. However, there is redundancy when we move from a disjunctive to a conjunctive form and vice versa \citep{wodes2002}. Throughout the current paper, we follow \eqref{impsys} for the analysis. It can be shown that all the analysis will still hold for the representation \eqref{second representation}. 

Now, we define the structure of the $A$, $B$, $C$ and $D$ matrices of an MMPS system with both time and quantity states as follows.
\begin{definition}
An autonomous implicit MMPS system with both temporal states and quantity states can be written as
\begin{equation} 
  \begin{aligned}
	\!\!\!\!
	     \matr{cc}{ \xt(k)\\ \xq(k)} 
		\!\!&=\!\! \underbrace{\matr{cc}{
			\At\!&\!\bm{\varepsilon}\\
			\bm{\varepsilon}\!&\!\Aq 	
		}}_{A}\!\!\otimes\!
		\Biggl(\!\underbrace{\matr{cc}{
			\Bt\!&\!\bm{\top}\\ \bm{\top}\!&\!\Bq
		}}_{B}\!\!\otimes'\!
		\biggl(\underbrace{\matr{cc}{
			C_{11}\!&\!C_{12}\\
			C_{21}\!&\!C_{22}
		}}_{C}\!\!\cdot\!\!
	    \matr{cc}{
		\xt(k-1)\\\xq(k-1)
	}\!\!  \\
	& \hspace*{2cm} 
	+\!\underbrace{\matr{cc}{
			D_{11}\!&\!D_{12}\\
			D_{21}\!&\!D_{22}
	}}_{D}\!\!\cdot\!\!
	\matr{cc}{
		\xt(k)\\\xq(k)
	}\biggl)\!\!\Biggl) ,
  \end{aligned}
  \label{eq:TQABCDsys}
\end{equation}
where 
$\xt \in \rrr^{n_\mathrm{t}}$, $\xq \in \rrr^{n_\mathrm{q}}$, 
$\At \in \Reps^{n_\mathrm{t} \times m_\mathrm{t}}$, $\Aq \in \Reps^{n_\mathrm{q} \times m_\mathrm{q}}$, $\Bt \in \Rtop^{m_\mathrm{t} \times p_\mathrm{t}}$, $\Bq \in \Rtop^{m_\mathrm{q} \times p_\mathrm{q}}$, 
$C_{11}, D_{11}\in \rrr^{p_\mathrm{t} \times n_\mathrm{t}}$, $C_{12}, D_{12}\in \rrr^{p_\mathrm{t} \times n_\mathrm{q}}$, $C_{21}, D_{21}\in \rrr^{p_\mathrm{q} \times n_\mathrm{t}}$, and $C_{22}, D_{22}\in \rrr^{p_\mathrm{q} \times n_\mathrm{q}}$. The notations $\bm{\varepsilon}$ and $\bm{\top}$ represent matrices of appropriate sizes with all elements equal to $\varepsilon$ and $\top$ respectively. 
\end{definition}
\begin{proposition}[Extended state MMPS system]
    An implicit MMPS system can be represented in the following extended state form:
    \begin{align}
    \label{extenmmps}
        \begin{aligned}
            x(k) &=A\otimes y(k),\\
            y(k) &= B\otimes'z(k),\\
            z(k) &= C\cdot x(k-1)+D\cdot x(k).
        \end{aligned}
    \end{align}
\end{proposition}
\begin{proof}
    By substituting $y(k)$, and $z(k)$ in $x(k)$, we directly get \eqref{impsys}.
\end{proof}
\begin{proposition}
\label{propC}
When an implicit MMPS system (\ref{impsys}) is time invariant, it satisfies the following properties, 
	\begin{align}
 \begin{aligned}
	   \sum_{i\in\overline{n_\mr{t}}} [C_{11}\; D_{11}]_{\ell i} \!=\! 1, \; \forall \ell \in \overline{p_\mr{t}},\;
	   \sum_{i\in \overline{n_\mr{t}}} [C_{21}\; D_{21}]_{ti} \!=\! 0, \; \forall t \in \overline{p_\mr{q}}.
    \end{aligned}
    \label{eq:propC}
    \end{align}	 
\end{proposition}
\begin{proof}
    Time-invariance implies that the MMPS system is additively homogeneous with respect to the temporal states. This means that if the temporal states $x_\mr{t}(k)$ are shifted by an amount $h\one$, all the extended temporal states \eqref{extenmmps} are shifted by the same amount. Therefore,
    \begin{align*}
        \xt(k)+h\one &= \At\otimes(y_\mr{t}(k)+h\one)=\At\otimes y_\mr{t}(k) + h\one \\
        y_\mr{t}(k)+ h\one &=  \Bt\otimes' (z_\mr{t}(k)+h\one)=\Bt \otimes' z_\mr{t}(k)+h\one \\
        z_\mr{t}(k)+h\one &= C_{11}\cdot (\xt(k-1)+h\one)+C_{12}\cdot \xq(k-1)\\&=+ D_{11}\cdot (\xt(k) + h\one)+D_{12}\cdot \xq(k).
    \end{align*}
    From \eqref{ti:prop}, the extended temporal state $z_\mr{t}(k)$ is additively homogeneous when 
    \begin{align*}
    z_\mr{t} (k)+h\one&= C_{11}\cdot \xt(k-1)+C_{12}\cdot \xq(k-1)+\\&\;\;\;\; D_{11}\cdot \xt(k) +D_{12}\cdot \xq(k)+ h\one.
    \end{align*}
    This is valid when
    $C_{11}\cdot h\one +D_{11} \cdot h\one = h\one$.
    Therefore, $\sum_{i\in\overline{n_\mr{t}}} [C_{11}\; D_{11}]_{\ell i} = 1, \; \forall \ell \in \overline{p_\mr{t}}$. The quantity states should not change with any shift in temporal states. From \eqref{eq:TQABCDsys} and \eqref{extenmmps}, only $z_\mr{q}(k)$ depends on the temporal states $\xt(k)$. So, when the temporal states are shifted,
\begin{align*}
x_\mr{q}(k) &= A_\mr{q}\otimes y_\mr{q}(k),\\
y_\mr{q}(k) &= B_\mr{q}\otimes'z_\mr{q}(k),\\
    z_\mr{q}(k) &= C_{21}\cdot (\xt(k-1)+h\one)+C_{22}\cdot \xq(k-1)\\&\;\;\;\; +D_{21}\cdot (\xt(k)+h\one) +D_{22}\cdot \xq(k).
\end{align*}
From \eqref{ti:prop}, the MMPS system is partially additively homogeneous when 
\begin{align*}
    z_\mr{q}(k) = f(x_\mr{t}+h\one, x_q) &= C_{21}\cdot \xt(k-1)+C_{22}\cdot \xq(k-1)\\&\;\;\;\; +D_{21}\cdot \xt(k) +D_{22}\cdot \xq(k).
\end{align*}
This is valid when $C_{21}\cdot h\one+D_{21}\cdot h\one =\! 0$. Hence, $\sum_{i\in \overline{n_\mr{t}}} [C_{21}\; D_{21}]_{ti} = \!0,~  \forall t \!\in \!\overline{p_\mr{q}}$.
\end{proof}
\begin{assumption}
    We assume that the MMPS systems in this paper are autonomous, regular (have regular system matrices $A$, $B$), time invariant, and have finite initial conditions ($x(0)\in \rrr^n$). 
\end{assumption}
The regularity and time invariance of max-plus algebraic systems are common assumptions in literature, \citep{olsder1994structural,mmps} as many applications of discrete-event systems automatically satisfy these properties. Finite initial conditions are required for the system to be well-defined, and is a reasonable assumption to have for real world applications. 
\section{Solvability of implicit MMPS systems}
\label{solvabilityimp}
In a number of practical applications, the implicit dependency
between some of the states occurs. For example, the
departure time of a train from a station $j$ in an urban
railway line depends on the arrival time of that train
in the station $j$ (see Section \ref{sec:urban}). Here, both these states
correspond to the event cycle $k$ but occur in a progressive
order. For complex applications, it may not be easy to
distinguish the order of states in an event cycle $k$. One
such example is when there is first come, first served
processing, where the sequence of incoming states are usually
not predetermined. The aim of this section is to determine
whether there exists an ordering of states such that each
state depends only on events that occurred earlier in the
same event cycle ‘$k$’. Such a dependency would ensure the
existence of an explicit expression for the dynamics of the
system. 

In general, an implicit system of the form \eqref{impsys} can be written in an explicit form.
The expression for this explicit MMPS system will be a nested MMPS system, 
	\begin{align} 
			x(k) &= A \otimes ( B \otimes' (C \cdot x(k-1) + D\cdot \nonumber \\
				 & (A \otimes ( B \otimes' \cdots  (C \cdot x(k-1) + D \cdot  \label{nestsys} \\
			     &\ldots (A \otimes ( B \otimes' (C \cdot x(k-1) )))  \cdots ))))) \nonumber.
	\end{align}	  
If this expression has a finite number of terms, we will have a solution for the state $x(k)$. However, when there are an infinite number of terms, we may not always find a solution. Rewriting an implicit MMPS system into an explicit form is a complex and difficult task. Also, rewriting the nested MMPS system \eqref{nestsys} with finitely many terms in the explicit ABC canonical form \eqref{abcmmps} may result in system matrices of impractically large size \citep{wodes2002}. Hence, the goal here is to prove the existence of an explicit form of the implicit MMPS system, but not to compute it. If such an explicit form exists, the implicit MMPS system is solvable and can be analyzed in its implicit form without increasing the computational complexity, as will be discussed in Section \ref{sectiongrowthrate}.

Now, consider the three structure matrices ($S_{A}$, $S_{B}$, $S_{D}$) as follows:
\begin{equation}
	\label{structuremat}
  \begin{aligned}
	&[S_\mr{A}]_{i,j} =  \begin{cases} 1 ~ \text{if} ~[A]_{i,j} \ne \eps \\
		                         0 ~ \text{if} ~[A]_{i,j} = \eps 
		            \end{cases}\hspace*{-1mm}
	[S_\mr{B}]_{i,j} =  \begin{cases} 1 ~ \text{if} ~[B]_{i,j} \ne \top \\
								 0 ~ \text{if} ~[B]_{i,j} = \top
					\end{cases} \!\!\!\!\!\! \\
	&[S_\mr{D}]_{i,j} =  \begin{cases} 1 ~ \text{if} ~[D]_{i,j} \ne 0 \\
								 0 ~ \text{if} ~[D]_{i,j} = 0
					\end{cases}.
  \end{aligned}
\end{equation}
The matrix product $S = S_\mr{A}\cdot S_\mr{B}\cdot S_\mr{D}$ determines the implicit dependency of the states $x_i(k), i\in \overline{n}$ to states $x_j(k), j\in \overline{n}$. When there exists an order in which events occur in a cycle `$k$', the matrix $S$ can be altered such that the state $x_i(k)$ implicitly depends only on states $x_j(k)$ where $j<i$. 
\begin{theorem}
	\label{th:unique}
 If there exists a permutation matrix $T\in \rrr^{n\times n}$ and its inverse in conventional plus-times algebra such that
	\begin{equation} 
		\label{eq:G}
		F = T\cdot S_A \cdot S_B \cdot S_{D}\cdot T^{-1}
	\end{equation}	
	is a strict lower triangle matrix, then there always exists a unique solution $x(k)$, $k>0$ for the implicit MMPS system (\ref{impsys}) for any state $x(k-1)$. 
\end{theorem}
    \begin{proof} 
    In this proof, we show that under a state transformation, it is possible to get an explicit form for the implicit MMPS system provided $F$ is strictly lower triangular. Existence of an explicit form implies the solvability of the implicit MMPS system. 
    
 Let $\Tx(k) = T \cdot x(k)$ be a state transformation, which is just a change of order of the state $x(k)$ as $T$ is a permutation matrix. Then, the new system dynamics become 
\begin{equation}
	\label{impsysG}
	\begin{aligned} 
		\Tx(k) 	&= f(\Tx(k),\Tx(k-1)) \\
				&= T\cdot( A \otimes ( B \otimes' (C\cdot T^{-1} \cdot \Tx(k-1) \\
				&  \hspace*{3cm}+ D\cdot T^{-1} \cdot \Tx(k) ))).
	\end{aligned}	
\end{equation}  
Consider the matrix $F$ as defined in \eqref{eq:G}. This matrix can represent a communication graph\footnote{A communication graph is a directed graph $\mathcal{G}$ which is a pair $(\mathcal{N},\mathcal{D})$ associated with a square matrix $F\in \R^{n\times n}$ where $\mathcal{N}$ is a finite set of nodes and $\mathcal{D}$ is a set of ordered pair of nodes called arcs.}, \citep{mpw} of the implicit MMPS system, intuitively without the weights. This means that the matrix $F$ represents a graph, where a non-zero $F_{i,j}$ tells us there is a connection between $\Tx_i(k)$ on the left side of \eqref{impsysG} and $\Tx_j(k)$ on the right side of \eqref{impsysG}, i.e.,

\begin{align*}
	\Tx_1(k) &= f_1(\Tx(k\!-\!1)) \\ 
	\Tx_i(k) &= f_i(\Tx_1(k),\hdots,\Tx_{i-1}(k),\Tx(k\!-\!1)), \;i\in \{2\hdots,n\}. 
\end{align*}
Successive substitution leads to:
\begin{align*}
	\Tx_1(k) &= f_1(\Tx(k\!-\!1)) \\
			 &= \bar{f}_1(\Tx(k\!-\!1)) \\ 
	\Tx_2(k) &= f_2(\bar{f}_1(\Tx(k\!-\!1)),\Tx(k\!-\!1)) \\
			 &= \bar{f}_2(\Tx(k\!-\!1)) \\
             \;\;\;\vdots & \;\;\;\;\;\vdots \\
             \Tx_n(k) &= f_n(f_1(\Tx(k\!-\!1)),\ldots,\Tx_{n\!-\!1}(k),\Tx(k\!-\!1))\\
			 &= \bar{f}_n(\Tx(k\!-\!1)).
		\end{align*}
This shows the existence of an explicit form:
\begin{align*}
	\Tx(k) = \bar{f}(\Tx(k\!-\!1))
\end{align*} 
and thus, 
\begin{align*}
	x(k) = T^{-1}\cdot \bar{f}(T\cdot x(k\!-\!1)),
\end{align*} 
which is an explicit MMPS function with finite terms.  
\end{proof}

\section{Time-invariant MMPS systems with multiple growth rates}
\label{sectiongrowthrate}
 The growth rate of an MMPS system is defined as the additive eigenvalue of the MMPS function that defines the system. A stable MMPS system should have the same growth rate for all of its temporal states and zero growth rate for all the quantity states, i.e., quantity states should be bounded for each event `$k$'. An MMPS system can have multiple temporal growth rates and fixed points based on the initial condition of the system. This is analogous to the existence of multiple equilibrium points for a nonlinear system in conventional algebra. In \citep{mmpsWODES2024}, we proposed a set of linear programming problems (LPPs) to find the growth rates and fixed points of an explicit MMPS system. 
Here, we extend this to find the growth rates and fixed points of an implicit autonomous MMPS system. 

Consider a time-invariant implicit MMPS system (\ref{impsys}) with 
$A \in \rrr^{n \times m}$, $B \in \rrr^{m \times p}$, and $C, D \in \rrr^{p \times n}.$ 
    The scalar $\lambda$ is a temporal growth rate of the system \eqref{extenmmps}, if there exists a fixed point $(\xte,\xqe,\yte,\yqe,\zte,\zqe)$ such that
\begin{align}
\begin{aligned}
\zte   &= C_{11} \cdot   (\xte-\lambda\one) + C_{12} \cdot \xqe +D_{11}\cdot \xte +D_{12}\cdot \xqe\\
\zqe   &= C_{21} \cdot   (\xte-\lambda\one) + C_{22} \cdot \xqe +D_{21}\cdot \xte +D_{22}\cdot \xqe\\
\yte   &= \Bt \otimes' \zte, \quad\yqe   = \Bq \otimes' \zqe \\
\xte   &= \At \otimes  \yte, \quad \xqe   = \Aq \otimes  \yqe .
\end{aligned}
\label{growthrateextndedstate}
\end{align}
Let $x_\mr{e} = [\xte^\intercal\;\xqe^\intercal]^\intercal$, $y_\mr{e} = [\yte^\intercal\;\yqe^\intercal]^\intercal$, $z_\mr{e} = [\zte^\intercal\;\zqe^\intercal]^\intercal$, and $s = [\one_{n_\mr{t}}^{\intercal}~\zero_{n_\mr{q}}^{\intercal}]^{\intercal}$. Define, $x_{\mr{e},\lambda}=x_\mr{e}-s\lambda$ and $A_\lambda=\begin{bmatrix}
    A_{\mr{t},\lambda}&\bm{\varepsilon}\\
    \bm{\varepsilon}&A_\mr{q}
\end{bmatrix}$, where $[A_{\mr{t},\lambda}]_{i_\mr{t}j_\mr{t}}=A_{i_\mr{t}j_\mr{t}}-\lambda,\;\forall i_\mr{t}\in \overline{n_\mr{t}},\;\forall j_\mr{t}\in \overline{m_\mr{t}}$. Then, from \eqref{growthrateextndedstate} 
\begin{align*} 
	z_e &= C \cdot x_{\mr{e},\lambda} + D \cdot (x_{\mr{e},\lambda}+s\lambda) \\
	&= D \cdot s \lambda  + (C + D) \cdot x_{\mr{e},\lambda}\\
	y_\mr{e} &= B \otimes' z_\mr{e} \\
	x_{\mr{e},\lambda} &= A_{\lambda} \otimes  y_\mr{e} .
\end{align*}
Let $w_e = z_e - D \cdot s \lambda$. Then
\begin{align} 
	w_\mr{e} &= (C+D) \cdot x_{\mr{e
},\lambda} \notag\\
	y_\mr{e} &= B \otimes' (z_\mr{e} - D \cdot s \lambda + D \cdot s \lambda )\notag\\
	y_\mr{e} &= B \otimes' (w_\mr{e}  + D \cdot s \lambda )\label{implicitalter}\\
	x_{\mr{e},\lambda} &= A_\lambda \otimes  y_e .\notag
\end{align}
Equation \eqref{implicitalter} can be written as
\begin{align*}
  y_\mr{e} &= (B + \lambda\one_m\cdot d^\intercal)\otimes'w_\mr{e},
\end{align*}
 where $d=D\cdot s$. Let $B' = B + \lambda\one_m\cdot d^\intercal$. Then
\begin{align} 
\begin{aligned}
	w_\mr{e} &= (C+D) \cdot x_{\mr{e},\lambda} \\
	y_\mr{e} &= B' \otimes' w_\mr{e} \\
	x_\mr{e} &= A_\lambda \otimes  y_\mr{e} .
    \end{aligned}
    \label{extendedfixedpoint}
\end{align}

Now define $X=\text{diag}_{\otimes}(x_{\mr{e},\lambda})$, $Y=\text{d}_{\otimes}(y_\mr{e})$, $Y'=\text{d}_{\otimes'}(y_\mr{e})$ and $W'=\text{d}_{\otimes'}(w_\mr{e})$.
Then 
\begin{align}
\begin{aligned}
	0 &= Y'^{-1} \otimes' y_\mr{e}\\
	&= Y'^{-1} \otimes' B' \otimes' W' \otimes' W'^{-1} \otimes' w_\mr{e} \\
	&= Y'^{-1} \otimes' B' \otimes' W' \otimes' 0 \\	
	&= \tilde{B} \otimes' 0\\
	0 &= X^{-1} \otimes x_{\mr{e},\lambda} \\
	&= X^{-1} \otimes A_\lambda \otimes Y \otimes Y^{-1} \otimes y_\mr{e} \\
	&= X^{-1} \otimes A_\lambda \otimes Y \otimes 0 \\	
	&= \tilde{A} \otimes 0 .
 \end{aligned}
 \label{elaboration_norm}
\end{align}
Therefore, the normalized MMPS system is
\begin{align}
    \begin{aligned}
\tilde{x}(k) &= \tilde{A} \otimes \big(\tilde{B}\otimes'(C\cdot \Tx(k-1)+D\cdot  \tilde{x}(k))\big),
\end{aligned}
\label{normalizedmmps}
\end{align}
where $\tilde{x}(k)$ is the normalized state. This system has a growth rate equal to $0$ and a fixed point equal to $\zero$. 
From \eqref{elaboration_norm}, we have $\zero = (C+D) \cdot\zero$, $	\zero = \tilde{B} \otimes' \zero $, and $	\zero = \tilde{A}  \otimes  \zero$. 
This means that $[\tilde{B}]_{j,\ell} \geq \zero$ and $[\tilde{A}]_{i,j} \leq \zero$. This shows that both $\tilde{A}$ and $\tilde{B}$ have a specific structure:
Each row of $\tilde{A}$ has at least one zero element, and all nonzero elements are less than zero; similarly, each row of $\tilde{B}$ has at least one zero element and all nonzero elements are greater than zero. So, from \eqref{elaboration_norm} 
\begin{align*} 
		[y]_j - [B']_{j\ell} - [w]_\ell &\leq 0, \\
		-[x]_i + [A]_{ij} -[s\lambda]_i+ [y]_j &\leq 0,
\end{align*} 
where $[B']_{j\ell} = [B]_{j\ell} + [d]_\ell\lambda$ and
    $d=D\cdot s$ as defined above. Then,
\begin{align*}
    [y]_j - [B]_{j\ell} -[d]_\ell\lambda- [w]_\ell &\leq 0 .
\end{align*}

The original MMPS states are related to the normalized states as follows: 
\begin{align}
\label{norm_orig_dependency}
\begin{aligned}
x(k) &= \tilde{x}(k) + (k \lambda)s + x_\mr{e},\\
y(k) &= \tilde{y}(k) + (k \lambda)s + y_\mr{e}, \\
z(k) &= \tilde{z}(k) + (k \lambda)s + z_\mr{e}  .
\end{aligned}
\end{align}
Let there be $S$ temporal growth rates for the implicit MMPS system, and let $\theta\in \{1,\hdots,S\}$. Then, we define a pair of footprint matrices $(G_{\mr{A}_{\theta}}, G_{\mr{B}_{\theta}})$ as follows:
\begin{align}
\label{footprintmat}
&~~~~G_{\mr{A}_{\theta}} = \begin{bmatrix}
G_{\mr{A}_{\mathrm{t}\theta}}& \mbox{\large$0$}\\
         \mbox{\large$0$} & G_{\mr{A}_{\mathrm{q}\theta}}
    \end{bmatrix},~~~~~~ G_{\mr{B}_{\theta}} = \begin{bmatrix}
        G_{\mr{B}_{\mathrm{t}\theta}}& \mbox{\large$0$}\\
         \mbox{\large$0$} &G_{\mr{B}_{\mathrm{q}\theta}}
    \end{bmatrix},\\
&[G_{\mr{A}_{\mathrm{t}\theta}}]_{ij} = \begin{cases} 1 ~ \text{if} ~[\tilde{A}_{\mathrm{t}\theta}]_{ij} = 0 \\ 0 ~\text{if}~[\tilde{A}_{\mathrm{t}\theta}]_{ij} < 0, \end{cases}\!\!\!\!\!\!
[G_{\mr{B}_{\mathrm{t}\theta}}]_{jl} \!=\! \begin{cases} 1~\text{if}~[\tilde{B}_{\mathrm{t}\theta}]_{jl} = 0 \\ 0 ~\text{if}~[\tilde{B}_{\mathrm{t}\theta}]_{jl} > 0, \end{cases}\notag\\
&[G_{\mr{A}_{\mathrm{q}\theta}}]_{rs} \!=\! \begin{cases} 1~\text{if} ~[\tilde{A}_{\mathrm{q}\theta}]_{rs}\! =\! 0 \\ 0 ~\text{if}~[\tilde{A}_{\mathrm{q}\theta}]_{rs}\! <\! 0,  \end{cases}\!\!\!\!\!\!
[G_{\mr{B}_{\mathrm{q}\theta}}]_{st}\! =\! \begin{cases} 1~\text{if}~[\tilde{B}_{\mathrm{q}\theta}]_{st} = 0 \\ 0 ~\text{if}~[\tilde{B}_{\mathrm{q}\theta}]_{st} > 0, \end{cases}\notag
\end{align}
where $i\in \overline{n_\mr{t}},\; j\in \overline{m_\mr{t}},\; l\in \overline{p_\mr{t}},\; r\in \overline{n_\mr{q}},\; s\in \overline{m_\mr{q}},\; t\in \overline{p_\mr{q}}$ and \mbox{\large$0$} denotes the zero matrix of appropriate size. Each pair of footprint matrix defines the location of zero entries in a normalized system. Hence, given $A$ and $B$ matrices of an MMPS system \eqref{abcmmps}, we can generate all possible combinations of footprint matrix pairs with exactly one `$1$' in each row of $G_{\mr{A}_\theta}$ and $G_{\mr{B}_\theta}$. Using the footprint matrix pairs \eqref{footprintmat}, we can generate all possible linear programming problems: $\forall i\in \overline{n}, \forall j \in \overline{m}, \forall \ell\in \overline{p}$,  the corresponding LPP is
\begin{align}
\label{alg1imp}
\begin{aligned}
\min_{x_\mr{e},y_\mr{e},w_{\mr{e}}} &\;\;\lambda\\
\text{s.t.}~~-&[s\lambda]_i-[x]_i + [y]_j  \leq -[A]_{ij}   ~~\text{if}~[G_{\mr{A}_\theta}]_{ij}=0\\ 
-&[s\lambda]_i-[x]_i+[y]_j =-[A]_{ij}  ~~\text{if}\;[G_{\mr{A}_\theta}]_{ij}=1\\
&[y]_j -[d]_\ell\lambda- [w]_\ell  \leq [B]_{j\ell}  ~~~\text{if} ~[G_{\mr{B}_\theta}]_{j\ell} =0 \\
&[y]_j -[d]_\ell\lambda- [w]_\ell  = [B]_{j\ell} ~~~\text{if} ~ [G_{\mr{B}_\theta}]_{j\ell} =1 \\
&d =D\cdot s,\;\;w = (C+D)\cdot x.
\end{aligned}
\end{align}
The inequality constraints guarantee that the entries of the normalized matrices satisfy  $[\TA_\theta]_{ij}<0$ when $[G_{A_\theta}]_{ij}=0$ and $[\TB_\theta]_{ij}>0$ when $[G_{B_\theta}]_{ij}=0$. The equality constraints make sure that at least one entry of each row of the normalized matrices $\TA, \; \TB$ is equal to $0$.  Each finite value of the $A$ and $B$ matrices give rise to either an equality constraint or an inequality constraint. A footprint matrix pair $(G_{A_\theta},G_{B_\theta})$ decides which of the finite entries in each row of $A$ and $B$ matrices give rise to the equality constraint. So, the total number of possible footprint matrix pairs depend on the number of finite entries in the $A$, and $B$ matrices. The number of LPPs is equal to the total number of possible footprint matrix pairs, which is less than or equal to $m_\mr{t}^{n_\mr{t}}p_\mr{t}^{m_\mr{t}}m_\mr{q}^{n_\mr{q}}p_\mr{q}^{m_\mr{q}}$. This is the same computational complexity as finding the growth rate and fixed points of an explicit MMPS system in \citet{mmpsWODES2024}.

The size of each LPP increases in a quadratic manner as the dimensions of system matrices $A, B, C, D$ increases. However, if any element in $A$ is $\varepsilon$ or $B$ is $\top$, the corresponding constraints can be omitted from the LPP as they give rise to constraint inequalities with $\top$ on the right-hand side of the equation. This reduces the computational complexity.
In particular, many practical systems (see the case study in Section \ref{sec:urban} for example) have a lot of $\varepsilon$ and $\top$ components in the matrices $A$ and $B$, respectively. Therefore, the number of possible footprint matrix pairs and hence the number of LPPs that give a solution drastically reduces. 
\begin{remark}
\label{lppred}
    Let $a_i$,
$i\in \overline{n}$ and $b_j$, $j\in \overline{m}$ be the number of finite elements in the $i$-th row of $A$ and $j$-th row of $B$, respectively. Then the number of possible footprint matrix pairs (i.e. the number of LPPs) is 
$\Pi_{i=1}^{n}a_i \cdot \Pi_{j=1}^{m}b_j$. This is still a conservative estimate, as not all LPPs yield a feasible growth rate. However, these results are a useful baseline, and we are actively exploring approaches to reduce the number of LPPs that gives a feasible solution.
\end{remark}

Let $(\lambda^*,v^*=[{x^*}^\intercal\;{y^*}^\intercal\;{w^*}^\intercal]^\intercal)$ be one of the solutions of problem \eqref{alg1imp}. The fixed point $v^*$ may not be the only solution that satisfies the inequality and the equality constraints of \eqref{alg1imp}. By substituting the value $\lambda^*$ in \eqref{alg1imp} we get a system of equality and inequality constraints, which is denoted as follows
\begin{align}
    H_{\text{eq}}\cdot v = h_{\text{eq}},\quad H_{\text{ineq}}\cdot v \leq h_{\text{ineq}}.
    \label{fixed pointfeasiblity}
\end{align}
For a specific $\lambda^*$, there could be a set of feasible solutions for \eqref{fixed pointfeasiblity}. Note that the matrix $H_{\text{eq}}$ is a square matrix, and the rank of $H_{\text{eq}}$ will always be at least one less than $n+m+p$. This is because the fixed point is shift-invariant in the direction of vector $s=[\one_{n_\mr{t}}^{\intercal}~\zero_{n_\mr{q}}^{\intercal}]^{\intercal}$.  This is directly related to the time-invariance property of the matrices $C$ and, $D$ as in  \eqref{eq:propC}. When the rank of the matrix $H_{\text{eq}}$ is less than $n+m+p-1$, there are more directions where the fixed points are shift-invariant. Hence, we get a set of fixed points corresponding to each $\lambda^*$. Let $v^*$ be a solution of \eqref{fixed pointfeasiblity} and let $\hat{s}_1,\hat{s}_2\hdots,\hat{s}_f\in \R^{n+m+p}$ be the set of extended direction vectors. As one of the direction vectors will always be $s$, let $\hat{s}_1$ be the extended direction vector corresponding to $s$. Then from \eqref{extendedfixedpoint}, $$\hat{s}_1 = [s^\intercal \;y_\mr{s}^\intercal\;w_\mr{s}^\intercal]^\intercal, y_\mr{s}= B\otimes' (C+D)\cdot s, \;w_\mr{s}=(C+D)\cdot s.$$ Then, the extended fixed points set $\mathcal{V}_{\lambda^*}$ of the system for a growth rate $\lambda^*$ is 
\begin{align}
    \mathcal{V}_{\lambda^*} &= \{v|H_{\text{ineq}}\cdot v\leq h_{\text{ineq}}\},
\end{align}
where $v=v^*+\sigma_1\hat{s}_1+\sigma_2\hat{s}_2+\hdots\sigma_f\hat{s}_f$. Then the fixed point set is the set of sub-vectors $x\in \R^n$ that consist of only the first $n$ components of $v\in \R^{n+m+p}$. Also, $$x = x^*+\sigma_1s+\sigma_2s_2+\hdots\sigma_fs_f,$$ where $s,s_2,\hdots,s_f\in \R^n$ are the direction vectors, which are the sub-vectors that consist of only the first $n$ components of the extended direction vectors $\hat{s}_1,\hat{s}_2,\hdots,\hat{s}_f\in \R^{n+m+p}$.  
As the fixed point is shift-invariant in the direction of $s$, the scaling factor $\sigma_1$ is unconstrained. This means that any vector in the direction of $s$ is a fixed point. This result directly follows from the time-invariance property. However, all the other vector directions $s_i,\; i>1$ may not be shift-invariant. Therefore, the scaling factors $\sigma_i,\; i>1$ may be constrained. The range of $\sigma_i,\; i>1$ can be found using the inequality constraints. In summary, the set of fixed points corresponding to a specific growth rate $\lambda^*$ will be a polyhedron, which is unconstrained in the direction of $s$.

\section{Local behavior of MMPS systems around fixed points}
\label{sec4}
In this section, the local behavior of an implicit MMPS system is examined around a fixed point that corresponds to a specific growth rate.
\begin{assumption}
\label{ass2}
   Throughout this section, we assume that there exists a permutation matrix $T$ such that
the matrix $F$ as defined in \eqref{eq:G} is a strictly lower-triangle matrix (see Theorem \ref{th:unique}).
\end{assumption}
Since events within a cycle `$k$' follow a specific order in many practical applications—including urban railway systems (see Section \ref{sec:urban}), manufacturing processes, and logistics—solvability is ensured. Hence, Assumption \ref{ass2} is well justified.

Consider an implicit system with multiple growth rates $\lambda_\theta$ represented by a set of 
 associated normalized MMPS systems
\begin{align}
\begin{aligned}
\tilde{x}_{\theta}(k) &= \tilde{A}_{\theta} \otimes \big(\tilde{B}_{\theta}\otimes'(C\cdot \Tx_\theta(k-1)+D\cdot  \tilde{x}_{\theta}(k))\big)
\end{aligned}
    \label{tildesysimp}
\end{align}
for ${\theta}\in\{1,\hdots,S\}$, and $\TA_{\theta}\in \rrr^{n\times m}$, $\TB_{\theta} \in \rrr^{m\times p}$, $C, D\in\rrr^{p\times n}$. Note that the normalized system matrices $\tilde{A} = X^{-1} \otimes A_\lambda \otimes Y ,\; \tilde{B}=Y'^{-1} \otimes' B' \otimes' W'$ are shift-invariant in the direction of vector $s=[\one_{n_\mr{t}}^{\intercal}~\zero_{n_\mr{q}}^{\intercal}]^{\intercal}$. Therefore, for a growth rate $\lambda_\theta$ and the set of fixed points $x_\mr{e}+hs,\;h\in \rrr$, we get the same normalized system.

Now, consider a region $\Omega_{\theta}$ containing all the vectors $x, w\in \rrr^n$ such that
\begin{align}
\label{omegaimp}
\TA_{\theta}\!\otimes\!\big(\TB_{\theta}\otimes'(C\!\cdot\! x\!+\!D\!\cdot\! w)\big)\! =\! G_{\mr{A}_{\theta}}\!\cdot\! G_{\mr{B}_{\theta}}\!\cdot\! (C\!\cdot\! x\!+\!D\!\cdot\! w).
\end{align}
This means that $y:=\TB_{\theta}\otimes'(C\cdot x+D\cdot w)$ selects the $i_{0,j}$-th component of the vector $(C\cdot x+D\cdot w)$ where $i_{0,j}$ denotes the position of a zero\footnote{There can be multiple zeros in same row of $A$ (or $B$). Please refer to Remark \ref{multizero} to see what happens in this case.} in row $j$ of $\TB_{\theta}$, and $\TA_{\theta}\otimes y$ selects the $l_{0,i}$-th component of the vector $y$ where $l_{0,i}$ denotes the position of a $0$ in row, $i$, of $\TA_{\theta}$. Note that $\Omega_\theta$ is a polyhedral region as \eqref{omegaimp} can be reduced to a system of inequalities. 
\begin{proposition}
\label{prop:linsys}
    Any normalized implicit MMPS system can be reformulated as a linear system in conventional algebra for all $\tilde{x}_{\theta}(k)\in \Omega_{\theta},\;k\in \zzz^+$ as follows:
    \begin{align}\label{linearsysimp}
        \begin{aligned}
            \Tilde{x}_{\theta}(k) &= M_{\theta}\cdot\Tilde{x}_{\theta}(k-1)\\
            M_{\theta}&=(I-M_1)^{-1}\cdot M_2 \\
            M_1 &=  G_{\mr{A}_{\theta}}\cdot G_{\mr{B}_{\theta}}\cdot D\\
            M_2 &=  G_{\mr{A}_{\theta}}\cdot G_{\mr{B}_{\theta}}\cdot C
        \end{aligned}
    \end{align}
    provided the inverse $(I-M_1)^{-1}$ exists. Hence, the linearized implicit MMPS system can be written into an explicit form. 
\end{proposition}
\begin{proof}
    The proof follows directly from \eqref{tildesysimp} and \eqref{omegaimp}.
\end{proof}
\begin{proposition}
When there exists a permutation matrix $T$ such that $F=T \cdot S_A \cdot S_B \cdot S_D \cdot T^{-1}$ is a strictly lower-triangle matrix, 
the inverse $(I-M_1)^{-1}$ always exists.
\end{proposition}
\begin{proof}
Note that the values of entries $[G_A]_{ij}$, $[S_A]_{ij}$, $[G_B]_{j\ell}$, $[S_B]_{j\ell}$, and $ [S_D]_{\ell i}$ are in $ \{0, 1\},\; i\in \overline{n}, j\in \overline{m}, \ell \in \overline{p}$.
As the matrices $G_A$, $G_B$ are constructed in such a way that each row of $G_A$, $G_B$ has exactly one 1 in it, there is
a correlation between matrices $G_A$, $G_B$ with matrices $S_A$, $S_B$. The matrices $G_A$, $G_B$ can be obtained by
eliminating a few 1's from each row of $S_A$, $S_B$, respectively. Therefore, 
\begin{align*}
    [S_A\cdot S_B]_{i\ell}\geq [G_A\cdot G_B]_{i\ell}\; \forall i,\ell
\end{align*}
and $[S_A\cdot S_B]_{i\ell} = 0 \implies [G_A\cdot G_B]_{i\ell} = 0$.
Hence, when 
\begin{align}
     &[S_A\cdot S_B \cdot S_D]_{iq} = 0,\; q\in \overline{n}\;\implies\;[S_A\cdot S_B]_{i\ell}\cdot [S_D]_{\ell q} =0 \;\forall \ell\notag\\
    &\implies \; [S_D]_{\ell q}=0, \; \forall \ell \;\text{such that}\; [S_A\cdot S_B]_{i\ell}\neq 0\notag\\
    &\implies \; [D]_{\ell q}=0,\; \forall \ell \;\text{such that}\;[S_A\cdot S_B]_{i\ell}\neq 0 \label{firsin}\\
     &\implies \; [D]_{\ell q}=0,\; \forall \ell \;\text{such that}\;[G_A\cdot G_B]_{i\ell}\neq 0\notag\\
     &\implies [G_A\cdot G_B\cdot D]_{iq} =0.\notag
\end{align}

As $T \cdot S_A \cdot S_B \cdot S_D \cdot T^{-1}$ is a strictly lower triangular matrix, from \eqref{firsin}, 
$  T \cdot G_A \cdot G_B \cdot D \cdot T^{-1} =  T \cdot M_1 \cdot T^{-1}$ is also
a strictly lower triangular matrix. This implies that the matrix 

$(I-T \cdot M_1 \cdot T^{-1}) = T \cdot (I-M_1) \cdot T^{-1}$ (Recall that $T$ is a permutation matrix) has full rank.
Since $T$ is a full rank square matrix, $(I-M_1)$ is also full rank and hence the inverse $(I-M_1)^{-1}$ exists.
\end{proof}
Note that the linearized system matrix $M_\theta$ is the same for all the vectors in the fixed point set corresponding to the growth rate $\lambda_\theta$. However, the region $\Omega_{\theta_i}$ around different fixed points $x_{\mr{e},i}$ will be different, except for the ones that are shifted along the vector $s$. For all the fixed points shifted along the vector $s$, the region $\Omega_\theta$ will be the same.   
\begin{proposition}
   Let $\Omega_{\theta_i}:=\{\Tx|H\cdot \Tx\leq h_i\}$ be the region around the fixed point $x_{\mr{e},i}$. Then, for a different fixed point $x_{\mr{e},j}$ from the set $\mathcal{V}_{\lambda_\theta}$ \eqref{fixed pointset} corresponding to the same growth rate $\lambda_\theta$, the region is $\Omega_{\theta_j}=\{\Tx|H\cdot \Tx\leq h_1+H\cdot(x_{\mr{e},i}-x_{\mr{e},j})\}$. 
\end{proposition}
\begin{proof}
    Consider the region around the fixed point $x_\mr{e_i}$,
    \begin{align*}
        \Omega_{\theta_i}:=\{\Tx|H\cdot \Tx\leq h_i\}.
    \end{align*}
    From \eqref{norm_orig_dependency}, we have
    \begin{align*}
        x &= \Tx+(k\lambda_\theta)s+x_{\mr{e},i}.\\
        \text{Hence,}\;\;H\cdot\Tx &= H\cdot x-H\cdot x_{\mr{e},i}-H\cdot (k\lambda_\theta)s.
    \end{align*}
    As the region $\Omega_{\theta_i}$ is unbounded in the direction of vector $s$ (\citep{mmpsWODES2024}, Proposition 5), we have $H\cdot (k\lambda_\theta)s=\zero$. So,
    \begin{align}
        H\cdot\Tx &= H\cdot x-H\cdot x_{\mr{e},i}\leq h_i\notag\\
        \text{and thus}\;\; H\cdot x &\leq h_i+H\cdot x_{\mr{e},i}.\label{shiftedregion}
    \end{align}
    Now consider the region $\Omega_{\theta_j}$ for the fixed point $x_{\mr{e},j}$  
    \begin{align*}
        H\cdot \Tx = H\cdot x-H\cdot x_{\mr{e},j}.
    \end{align*}
    From \eqref{shiftedregion}, we get
    \begin{align*}
        H\cdot x-H\cdot x_{\mr{e},j}\leq h_i+H\cdot x_{\mr{e},i}-H\cdot x_{\mr{e},j}.
    \end{align*}
\end{proof}
\begin{proposition}
\label{propomegatheta}
    The region $\Omega_{\theta}$ associated to the implicit linearized system \eqref{linearsysimp} for a fixed point $x_\mr{e}$ is given by the system of inequalities
    \begin{align*}
           &H\cdot \Tx\leq h,\;\; H = \begin{bmatrix}
               U\\-L
           \end{bmatrix},\;\; h = \begin{bmatrix}
               \tilde{b}\\ -\tilde{a}
           \end{bmatrix}\\
           &U \!=\!\big((G_{\mr{B}_{\theta}}\boxtimes\one_{p})-(\one_{m}\boxtimes I_{p})\big)\cdot (C+D\cdot M_\theta)\\ 
           &L \!=\! \big((G_{\mr{A}_{\theta}}\boxtimes\one_{m})-(\one_{n}\boxtimes I_{m})\big)\!\cdot\! G_{\mr{B}_{\theta}}\!\cdot\! (C+D\cdot M_\theta)\\
           &\tilde{b}\!=\!\text{vec}(\TB_{\theta}),\quad\tilde{a}\! =\!
               \text{vec}(\TA_{\theta}),
    \end{align*}
    where $\Tx\in\rrr^n$,  $\boxtimes$ is the Kronecker product (see Definition \ref{kronecker}) and $\text{vec}(\cdot)$ is the vector constructed in the row major stacking (see Definition \ref{rowmajor}) of the matrix.
\end{proposition}
\begin{proof}
    The condition \eqref{omegaimp} can be reformulated as
         \begin{align}
             \TB_{\theta} \otimes' \tilde{z} &= G_{\mr{B}_{\theta}}\cdot \tilde{z}, \quad \tilde{z} = C\cdot \Tx(k-1)+D\cdot \Tx(k) \label{equvalencycond1}\\
             \TA_{\theta} \otimes \tilde{y} &= G_{\mr{A}_{\theta}} \cdot \tilde{y}, \quad \tilde{y} =G_{\mr{B}_{\theta}}\cdot \tilde{z} \label{equvalencycond2}
             \end{align}
             for some $\tilde{z}\in \rrr^p$, $\tilde{y}\in \rrr^m$, $\Tx\in \rrr^n$.
             From \eqref{equvalencycond1} we have
             \begin{align*}
             \min_{k\in \overline{p}}([\TB_{\theta}]_{jk}+\tilde{z}_k) &=[G_{\mr{B}_{\theta}}\cdot \tilde{z}]_j, \quad j\in \overline{m}.
             \end{align*}
             So, $[\TB_{\theta}]_{jk}+ \bm{e}_k \cdot \tilde{z} \geq [G_{\mr{B}_{\theta}}]_j\cdot \tilde{z},\;\forall j,k$ and thus
             \begin{align}
             ([G_{\mr{B}_{\theta}}]_j-\bm{e}_k)\cdot \tilde{z} &\leq [\TB_{\theta}]_{jk}, \quad\forall j,k.
             \label{U}
             \end{align}
             Here, $\bm{e}_k$ is the standard basis vector (see Section \ref{sec2}). Therefore, \eqref{U} can be written as
             
             $(G_{\mr{B}_{\theta}}\boxtimes\one_p-\one_m\boxtimes I_p)\cdot (C \cdot \Tx(k-1)+D\cdot \Tx(k)) \leq \text{vec}(\TB_{\theta})$
        From \eqref{linearsysimp} $\Tx(k) = M_\theta \Tx(k-1)$. Then,
        $$(G_{\mr{B}_{\theta}}\boxtimes\one_p-\one_m\boxtimes I_p)\cdot (C+D\cdot M_\theta)\cdot \Tx(k-1)) \leq \text{vec}(\TB_{\theta}).$$
        
         Similarly, from \eqref{equvalencycond2}, we have
         \begin{align}
             \max_{j\in \overline{m}}([\TA_{\theta}]_{ij}+y_j) &= [G_{\mr{A}_{\theta}}\cdot y]_{i}, \quad i\in \overline{n} \notag\\
             [\TA_{\theta}]_{ij}+\bm{e}_j\cdot y &\leq [G_{\mr{A}_{\theta}}]_i\cdot y, \quad \forall i,j \notag\\
             ([G_{\mr{A}_{\theta}}]_i-\bm{e}_j)\cdot y &\geq [\TA_{\theta}]_{ij}, \quad \forall i,j.
             \label{L}
         \end{align}
         Equation \eqref{L} can be written as
         
             $(G_{\mr{A}_{\theta}}\boxtimes\one_m-\one_n\boxtimes I_m)\cdot G_{\mr{B}_{\theta}}\cdot (C+D\cdot M_\theta) \cdot \Tx \geq \text{vec}(\TA_{\theta})$.
\end{proof}
\begin{remark}
\label{multizero}
    When there are multiple zeros in the same row of $\TA_\theta$ ($\TB_\theta$), there exists multiple footprint matrix pair corresponding to a growth rate $\lambda_\theta$ that can be obtained from \eqref{alg1imp}. The existence of multiple footprint matrix points to the possibility of different polyhedral regions around the fixed point. In this case, the fixed point lies on the boundary of these regions.
\end{remark}

\begin{proposition}
    The linearized implicit system matrix $M_\theta$ has the following properties:
    \begin{align*}
        M_\theta\cdot s = s, \;\;\text{with} \;\; s=[\one_{n_{\mr{t}}}^{\intercal} ~ \zero_{n_{\mr{q}}}^{\intercal}]^{\intercal}.
    \end{align*}
\end{proposition}
\begin{proof}
    From Proposition \ref{propC}, we have 
    $(C+D)\cdot s = s$. The product of the footprint matrices, $G_{\mr{A}_{\theta}}\cdot G_{\mr{B}_{\theta}}$  is a selection matrix with exactly one $1$ in each row. Hence, the matrices $M_1$ and $M_2$ consist of copies of rows of the matrices $D$ and $C$. So, we have
    \begin{align}
        (M_1+M_2)\cdot s &= s\notag\\
        M_2\cdot s &= s-M_1\cdot s \label{Ms}\\
        M_\theta\cdot s &= X\cdot M_2\cdot s,\notag
    \end{align}
    where $X = I-M_1$. Then, from \eqref{Ms}
    \begin{align}
        M_\theta\cdot s&=X\cdot(s-M_1\cdot s).
        \label{first}
    \end{align}
    Now we have 
    \begin{align}
        X\cdot (I-M_1)^{-1}&= I\notag\\
        X&= I+X\cdot M_1\notag\\
        X\cdot s&= s+X\cdot M_1\cdot s.
        \label{second}
    \end{align}
    Therefore, from \eqref{first} and \eqref{second},
    \begin{align*}
        X\cdot s-X\cdot M_1\cdot s = s = M_\theta\cdot s.
    \end{align*}
\end{proof}
So, the matrix $M_{\theta}$ will have at least a multiplicative eigenvalue equal to $1$ with eigenvector $v_1\!=\![\one_{n_{\mr{t}}}^{\intercal} ~ \zero_{n_{\mr{q}}}^{\intercal}]^{\intercal}$,
i.e., $M_{\theta}\cdot v_1\!=\! 1\cdot v_1$.

The local max-plus bounded buffer stability of an MMPS system \eqref{linearsysimp} can be assessed using the stability criteria of discrete-time systems \citep{hespanha2018linear} in the conventional algebra, as stated in the following Proposition.
    \begin{proposition}
    \label{propstability}
        The linearized system \eqref{linearsysimp} for ${\theta} \in \{1,\hdots,S\}$ is 
        \begin{itemize}
            \item max-plus bounded buffer stable if the system matrix $M_{\theta}$ has multiplicative eigenvalues less than or equal to one and the multiplicative eigenvectors corresponding to the multiplicative eigenvalues with magnitude one are independent. 
            \item  unstable if it has at least one multiplicative eigenvalue greater than one or the multiplicative eigenvectors associated with the multiplicative eigenvalues with magnitude one are not independent.
        \end{itemize}
    \end{proposition}
    Let $\Tx_0\in \Omega_\theta$ be an initial condition for a stable linearized system \eqref{linearsysimp}, $\mu_i$ be the set of multiplicative eigenvalues and $v_i$ be the eigenvectors of the linear system matrix $M_\theta$. Then,
    \begin{align*}
    x(0)&= \sigma_1v_1+\sigma_2v_2+\cdots+\sigma_nv_n,\;\;\sigma_i\in \R\\
        x(1)=M\cdot \Tx_0 &= \sigma_1\mu_1v_1+\sigma_2\mu_2 v_2+\cdots+\sigma_n\mu_n v_n.
    \end{align*}
    Without loss of generality, we can take $\mu_1 = 1$ and $v_1 = s$. Then, 
    \begin{align*}
         x(1) = M\cdot \Tx_0 &= \sigma_1s+\sigma_2\mu_2 v_2+\cdots+\sigma_n\mu_n v_n.\\
    \end{align*}
    If all $|\mu_i|<1,~\forall i>1$, as $k\to\infty$, we have 
    \begin{align*}
        x(k) = \sigma_1s.
    \end{align*}
    This means that the linear system \eqref{linearsysimp} is asymptotically stable and all the states converges to the vector $s$. When there are more than one multiplicative eigenvalues equal to $1$ (say $\mu_1, \mu_2$) with independent eigenvectors $s_1,\;s_2$ (where $s_1 = s$), then as $k\to\infty$
    \begin{align*}
        x(k)&= \sigma_1s_1+\sigma_2s_2.
    \end{align*}
    This shows that the states of the linear system \eqref{linearsysimp} can converge to the set of equilibrium points on a plane defined by the vectors $s_1, s_2$. This is directly related to the rank of the matrix $H_\text{eq}$.  
    \begin{corollary}
    \label{corrankeigrelation}
        The number of multiplicative eigenvalues of the matrix $M_\theta$ that are equal to $1$ for a stable linear system \eqref{linearsysimp} is equal to the rank deficiency of the matrix $H_\text{eq}$, i.e., $(n+m+p)-\text{rank}(H_\text{eq})$. This is equal to the dimension of the set of fixed points of the MMPS system \eqref{impsys} for a specific growth rate $\lambda_\theta$.
    \end{corollary}
\section{Case Study : An urban railway line} \label{sec:urban}
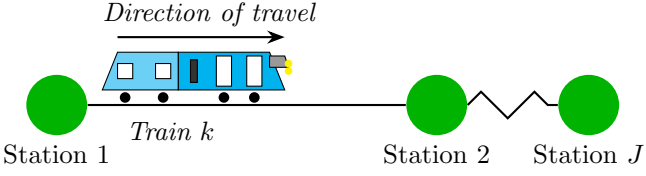
\begin{figure}
    \centering
    \begin{tikzpicture}[>=Stealth]

\node[circle, fill=green!70!black, minimum size=0.8cm, label=below:Station $1$] (S1) at (0,0) {};
\node[circle, fill=green!70!black, minimum size=0.8cm, label=below:Station $2$] (S2) at (5,0) {};
\node[circle, fill=green!70!black, minimum size=0.8cm, label=below:Station $J$] (SJ) at (7,0) {};

\draw[thick] (S1) -- (S2);
\draw[thick, decorate, decoration={zigzag, amplitude=5pt, segment length=20pt}] (S2) -- (SJ);

\begin{scope}[shift={(0.6,0.2)}]

  \draw[fill=cyan!50] (0,0) -- (0.2,0.5) -- (1,0.5)--(1,0)--cycle;

  \draw[fill=cyan!90] (1,0) -- (1,0.5) -- (2.2,0.5) -- (2.4,0) -- cycle;

  \foreach \x in {0.2,0.7} {
    \draw[fill=white] (\x,0.15) rectangle +(0.2,0.2);
  }

  \foreach \x in {1.5,1.9} {
    \draw[fill=white] (\x,0.05) rectangle +(0.2,0.4);
  }

  \draw[fill=black!80] (1.16,0.1) rectangle (1.25,0.4);
  
  \foreach \x in {0.3, 0.8, 1.6, 2} {
    \fill[black] (\x,-0.1) circle (0.07);
  }

  \draw[fill=black!40] (2.2,0.3) -- (2.2,0.45) -- (2.4,0.45) -- (2.45,0.3) -- cycle;

  \fill[yellow] (2.45,0.25) circle (0.05);
  \fill[yellow] (2.45,0.35) circle (0.05);
\end{scope}

\node[font=\itshape] at (2,1.2) {Direction of travel};
\draw[->, thick] (0.8,0.9) -- (3,0.9); 
\node[below] at (1.5,-0.1) {\textit{Train $k$}};
\end{tikzpicture}
\caption{An urban railway line} 
	\label{fig:rail}
\end{figure}
In this section, we present a case study of an urban railway line modelled as an implicit MMPS system. This example is taken from \citep{mmps}. 
Consider an urban railway line as given in Figure \ref{fig:rail} with $J$ station and $K$ trains. We assume there is no timetable. The trains indexed by $k\in\{1,\ldots,K\}$ depart from station 1 with a headway interval $\tau_0$. They stop at each station $j\in\{1,\ldots,J\}$, and depart from the station if all passengers have disembarked and boarded the train. We denote the arrival and departure time of the train $k$ at the station $j$ by $a_j(k)$ and $d_j(k)$, respectively. Let $\tau_\mr{d}$ be the dwell time of train $k$ at station 1. We denote the number of passengers in train $k$ when it leaves station $j$ by $\rho_j(k)$, and  
the number of passengers at station $j$ when train $k$ is leaving the station by $\sigma_j(k)$.   
We assume that all the trains have a limited capacity of $\rho_{\max}$ passengers. 
Also, consider the running times $\tau_{\mr{r},j}$ from station $j-1$ to station $j$ to be fixed.
Let $e_j$ denote the number of passengers entering the platform at station $j$ per time unit. Let $b$ be the number of passengers that can board the train per time unit, and let $f$ denote the number of passengers that can disembark the train per time unit. We assume that $b>e_j$ for all $j$ and that the number of passengers leaving train $k$ at a particular station $j$ is a fixed fraction $\beta_j$ of the number of the passengers in the train when it enters the station $j$.
The minimum headway between train $k$ and $k-1$ for all $k$ is given by $\tau_\mr{H}$. 

Finally, we define the parameters
\begin{align*}
    \mu_1 &= \frac{b}{b-e_j},\; \mu_2=\frac{b}{b-e_j}\frac{\beta_j}{f}, \;\mu_3= \frac{1}{b-e_j},\\
    \gamma_1 &= \frac{1}{b}\rho_{\max},\;\gamma_2 = \frac{\beta_j}{f}- \frac{1-\beta_j}{b}.
\end{align*}


The system equations for the urban railway line can now be derived. 
For $j>1$ and $k>0$ we obtain the following MMPS model \citep{mmps}: 
\begin{equation}\label{eq:urbanfinalA}
\begin{aligned}
	\begin{array}{rl} 
		a_j(k)			&= \max\Big( d_{j-1}(k) + \tau_{r,j}\>,\> d_j(k-1)+\tau_H \Big) \\
		d_j(k)       	&= \min\Big( \mu_1 a_j(k) + \mu_2 \rho_{j-1}(k) + \mu_3 \sigma_j(k-1)  +\\
		&  (1-\mu_1) d_j(k-1) \;,\;\gamma_1 + a_j(k) + \gamma_2 \rho_{j-1}(k) \Big)  \\[5mm]
		\rho_j(k)   	&= (1\!-\!\beta_j) \rho_{j\!-\!1}(k) + b\big(d_j(k) - a_j(k) -\frac{\beta_j}{f}\rho_{j\!-\!1}(k)\big)  \\
		\sigma_j(k)   	&= \sigma_j(k-1) + e_j \big(d_j(k)-d_j(k-1)\big)   \\
		&\hspace*{5mm}   - b\big(d_j(k) - a_j(k) -\frac{\beta_j}{f}\,\rho_{j-1}(k)\big). 
	\end{array}\hspace*{-11mm}{} 
 \end{aligned}
\end{equation}
For $j=1$ we get: 
\begin{align*} 
	a_1(k)      &= a_1(k-1) + \tau_0 & a_1(0)=0              \\
	d_1(k)      &= a_1(k)   +\tau_\mr{d}      & d_1(0)=\tau_\mr{d}         \\
	\rho_1(k)  &= \rho_1(k-1)           & \rho_1(0)=\bar{\rho}_1\leq \rho_{\max}\\
	\sigma_1(k) &=0 .
\end{align*}
We choose the states 
$x_{1,j}(k)=a_j(k)$, $x_{2,j}(k)=d_j(k)$, $x_{3,j}(k)=\rho_j(k)$, and $x_{4,j}(k)=\sigma_j(k)$. 
Now define the following matrices	for $j>1$: 
\[ \bar{A}_j= \matr{cccccc}{
               \tau_{\mr{r},j}  & \tau_\mr{H} & \eps   & \vline & \eps   & \eps \\
               \eps        & \eps   & 0      & \vline & \eps   & \eps \\
               \hline\\[-7mm]
                           &        &        & \vline &        &          \\
               \eps        & \eps   & \eps   & \vline & 0      & \eps \\               
               \eps        & \eps   & \eps   & \vline & \eps   & 0    }  \;,\;\;\;
   \bar{B}_j= \matr{ccccccc}{
	0    & \top & \top & \top     & \vline   & \top & \top \\
	\top & 0    & \top & \top     & \vline   & \top & \top \\
	\top & \top & 0    & \gamma_1 & \vline   & \top & \top \\
	\hline\\[-7mm]
	     &      &      &          & \vline   &      &          \\
	\top & \top & \top & \top     & \vline   & 0    & \top \\
	\top & \top & \top & \top     & \vline   & \top & 0    }  \]
\[	\bar{C}_{j} = \matr{cccccc}{
	0     & 0          & \vline & 0          & 0        \\
	0     & 1          & \vline & 0          & 0        \\
	0     & (1-\mu_1)  & \vline & 0          & \mu_3    \\
	0     & 0          & \vline & 0          & 0        \\
	\hline\\[-7mm]
	&       & \vline &                 &          \\
	0     & 0          & \vline & 0 & 0        \\
	0     & -e_j       & \vline & 0          & 1        }
\]
\[ \bar{D}_{j,\mr{c}} = \matr{cccccc}{
	0     & 0     & \vline & 0          & 0        \\
	0     & 0     & \vline & 0          & 0        \\
	\mu_1 & 0     & \vline & 0          & 0        \\
	1     & 0     & \vline & 0          & 0          \\
	\hline\\[-7mm]
	&       & \vline &            &          \\
	-b     & b     & \vline & 0          & 0         \\
	b     & e_j-b & \vline & 0          & 0        }  \;,\;\;\; 
\bar{D}_{j,\mr{p}} = \matr{cccccc}{
	0     & 1     & \vline & 0           & 0         \\
	0     & 0     & \vline & 0           & 0        \\
	0     & 0     & \vline & \mu_2       & 0        \\
	0     & 0     & \vline & \gamma_2    & 0          \\
	\hline\\[-7mm]
	&       & \vline &             &          \\
	0     & 0     & \vline & -b\gamma_2 & 0        \\
	0     & 0     & \vline & b\beta_j/f  & 0        },\]
where the matrix $D_{j,\mr{c}}$ is associated with the states at the current station $j$ and the matrix $D_{j,\mr{p}}$ is associated with states in the previous station $j-1$. For the starting station, i.e., $j=1$, we define
\[ \bar{A}_1= \matr{cccccc}{
	0           & \eps   & \eps   & \vline & \eps   & \eps \\
	\eps        & \eps   & 0      & \vline & \eps   & \eps \\
	\hline\\[-7mm]
	            &        &        & \vline &        &          \\
	\eps        & \eps   & \eps   & \vline & 0      & \eps \\               
	\eps        & \eps   & \eps   & \vline & \eps   & 0    }  \;,\;\;\;					
   \bar{B}_1= \matr{ccccccc}{
	\tau_0 & \top & \top               & \top     & \vline   & \top & \top \\
	\top       & 0    & \top               & \top     & \vline   & \top & \top \\
	\top       & \top & \tau_\mr{d}             & \top     & \vline   & \top & \top \\
	\hline\\[-7mm]
	           &      &                    &          & \vline   &      &          \\
	\top       & \top & \top               & \top     & \vline   & 0    & \top \\
	\top       & \top & \top               & \top     & \vline   & \top & 0    }  \]
\[\bar{C}_{1} = \matr{cccccc}{
	1     & 0     & \vline & 0          & 0        \\
	1     & 0     & \vline & 0          & 0        \\
	0     & 0     & \vline & 0          & 0        \\
	0     & 0     & \vline & 0          & 0        \\
	\hline\\[-7mm]
	      &       & \vline &            &          \\
	0     & 0     & \vline & 1          & 0        \\
	0     & 0     & \vline & 0          & 0        }
\;,\;\;\;
\bar{D}_{1,\mr{c}} = \matr{cccccc}{
	0     & 0     & \vline & 0          & 0        \\
	0     & 0     & \vline & 0          & 0        \\
	1     & 0     & \vline & 0          & 0        \\
	1     & 0     & \vline & 0          & 0        \\
	\hline\\[-7mm]
	      &       & \vline &            &          \\
	0     & 0     & \vline & 0          & 0        \\
	0     & 0     & \vline & 0          & 0        }  .
\]
By defining 
\begin{align} 
\begin{aligned}
	 A&=  \matr{cccc}{
		\bar{A}_{1} & \eps        & \cdots & \eps        \\
		\eps        & \bar{A}_{2} &        & \vdots      \\
		\vdots      &             & \ddots & \eps        \\
	    \eps        & \cdots      & \eps   & \bar{A}_{J} }  \;,\;\;\;
	 B=		\matr{ccccc}{
		\bar{B}_{1} & \top        & \cdots & \top        \\
		\top        & \bar{B}_{2} &        & \vdots      \\
		\vdots      &             & \ddots & \top        \\
		\top        & \cdots      & \top   & \bar{B}_{J} },   			\\	
C&=	   \matr{ccccc}{   
	   \bar{C}_{1} & \mbox{\large$0$}           & \cdots  & \mbox{\large$0$}             \\
	   \mbox{\large$0$}            & \bar{C}_{2} &         & \vdots             \\
	   \vdots        &               & \ddots  & \mbox{\large$0$}        \\
	  \mbox{\large$0$}          & \cdots        & \mbox{\large$0$}      & \bar{C}_{J} } ,\;\;	D=	    \matr{ccccc}{
			\bar{D}_{1,\mr{c}} & \mbox{\large$0$}             & \cdots & \mbox{\large$0$}             \\
			\bar{D}_{2,\mr{p}} & \bar{D}_{2,\mr{c}} &        & \vdots        \\
			\mbox{\large$0$}             & \ddots        & \ddots & \mbox{\large$0$}             \\
			\mbox{\large$0$}            & \cdots        & \bar{D}_{J,\mr{p}} & \bar{D}_{J,\mr{c}} },
            \end{aligned}
            \label{impliciturs}
\end{align} 
we obtain the implicit MMPS model \eqref{impsys}. 
To check if this implicit model is solvable, we first check the solvability condition in Theorem \ref{th:unique} using the structure matrices of $A, B$, and $D$. The structure matrices $S_\mr{A}, S_\mr{B}, S_\mr{D}$ are generated using \eqref{structuremat}. For ease of analysis, consider an urban railway line with 2 stations. Then, the structure matrices for this system have the following form:
\begin{align*}
    S_\mr{A} = \begin{bmatrix}
        S_\mr{\bar{A}_1}&\mbox{\large$0$}\\
        \mbox{\large$0$}& S_{\mr{\bar{A}}_j}
    \end{bmatrix},\; S_\mr{B} = \begin{bmatrix}
        S_\mr{\bar{B}_1}&\mbox{\large$0$}\\
        \mbox{\large$0$}& S_{\mr{\bar{B}}_j}
    \end{bmatrix},\; S_\mr{D} = \begin{bmatrix}
        S_\mr{\bar{D}_{1,c}}&\mbox{\large$0$}\\
        S_{\mr{\bar{D}}_{j,\mr{p}}}& S_{\mr{\bar{D}}_{j,\mr{c}}}
    \end{bmatrix}.
\end{align*}
 The structure matrices $S_\mr{\bar{A}_1}, S_{\mr{\bar{A}}_j}, S_\mr{\bar{B}_1}, S_{\mr{\bar{B}}_j}$ are obtained by replacing all the finite elements in the corresponding system matrices with $1$ and others with $0$ whereas $S_\mr{\bar{D}_{1,c}}, S_{\mr{\bar{D}}_{j,\mr{p}}}, S_{\mr{\bar{D}}_{j,\mr{c}}}$ is obtained by replacing all the non-zero elements in $\bar{D}_{j,\mr{p}}, \bar{D}_{j,\mr{c}}$ with $1$ and others with $0$ as in \eqref{structuremat}. Consider the matrix product $S = S_\mr{A}\cdot S_\mr{B}\cdot S_\mr{D}$,
 \begin{align*}
     S = \begin{bmatrix}
         S_\mr{\bar{A}_1}\cdot S_\mr{\bar{B}_1}\cdot S_\mr{\bar{D}_{1,c}}&\mbox{\large$0$}\\
         S_{\mr{\bar{A}}_j}\cdot S_{\mr{\bar{B}}_j}\cdot S_{\mr{\bar{D}}_{j,\mr{p}}}&
         S_{\mr{\bar{A}}_j}\cdot S_{\mr{\bar{B}}_j}\cdot S_{\mr{\bar{D}}_{j,\mr{c}}}
     \end{bmatrix}.
 \end{align*}
 The matrix products $S_\mr{\bar{A}_1}\cdot S_\mr{\bar{B}_1}\cdot S_\mr{\bar{D}_{1,c}}$ and $S_{\mr{\bar{A}}_j}\cdot S_{\mr{\bar{B}}_j}\cdot S_{\mr{\bar{D}}_{j,\mr{c}}}$ are strictly lower triangular. So, without loss of generality, we can conclude that $S$ is always strictly lower triangular for any number of stations. Therefore, the implicit model of the urban railway line is solvable according to Theorem \ref{th:unique}. The permutation matrix $T$ in this case is the identity matrix $I_n$.
 
Now, consider the following chosen parameters for an example urban railway network: 
\begin{align*}
	J         &=4,		           & \rho_{\max}   &=150,						& \beta_j 	&= 0.5, 			\\
	\tau_0&=120 ,               & \tau_\mr{r}		   &=120,						& \tau_\mr{H}	&= 30 , 			& \tau_\mr{d}		   &=60	,					\\		
	b         &= 2,		           & e_j             &= 0.5  ,						& f 		&= 2  . 			
\end{align*}
With these parameters, we get
\begin{align*} 				
	\mu_1	 &= 4/3,    	& \mu_2		&= 1/3,  	& \mu_3		&= 2/3,\\ 
	\gamma_1 &= 75, 		& \gamma_2 	&= 0. 
\end{align*}
We first calculate global parameters such as the growth rate and the fixed points of the urban railway system using LPPs (20). As mentioned in Remark 1, the number of LPPs that have to be solved for an urban railway line is dependent on the number of finite terms in each row of the $A$ and $B$ matrices. For the first station, there is only one finite element in each row of $\bar{A}_1$ and $\bar{B}_1$ and hence only one possible pair $(G_\mr{\bar{A}_1}, G_\mr{\bar{B}_1})$ of footprint matrices. Each of the corresponding stations,  $\bar{A}_j$ and $\bar{B}_j$, together give rise to $2\times2 = 4$ possible pairs of footprint matrices $G_{\mr{\bar{A}}_j}, G_{\mr{\bar{B}}_j}$ according to Remark \ref{lppred}. The footprint matrices of an urban railway line with four stations ($J=4$) have the following form:
  \begin{align*}
      G_\mr{A} = \begin{bmatrix}
          G_\mr{\bar{A}_1}&\cdots&\mbox{\large$0$}\\
          \vdots&\ddots\\
        \mbox{\large$0$}&\cdots& G_\mr{\bar{A}_4}
      \end{bmatrix},\; G_\mr{B} = \begin{bmatrix}
          G_\mr{\bar{B}_1}&\cdots&\mbox{\large$0$}\\
          \vdots&\ddots\\
        \mbox{\large$0$}&\cdots& G_\mr{\bar{B}_4}
      \end{bmatrix}.
  \end{align*}
  Hence, there are $1\times 4\times 4\times 4=64$ total different footprint matrix pairs that correspond to 64 different LPPs to be solved for this example. These LPPs (20) were solved with MATLAB R2024b on an Intel Core i7-1365U processor, and it took $7.45 \sec$  to solve all of them. For $n$ stations, there will be $4^{n-1}$ LPPs that should be solved to find all the growth rates of the system.  However, out of 64 LPPs, only one pair of footprint matrices yields a growth rate and all the other 63 LPPs are infeasible. So this algorithm is conservative and further research is going on to reduce the number of LPPs that give a feasible solution.
  
  The values of the growth rate and one of the fixed points for the urban railway line with $4$ stations are,
\begin{align}
\lambda&= 120,\\
x_{\mr{e_1}}\!& =\! \begin{array}{cccccccccccccccc}
   \big[0&  60 &   120&  0& 180& 240& 120&   0\hdots\\ \hdots360 & 420 &  120&   0 & 540  & 600 &   120&    0\big]^\intercal.
\end{array}
\label{URSpara}
\end{align}
Note that we provide only $x_\mr{e}$ of the fixed point $v = [x_\mr{e}^\intercal\;y_\mr{e}^\intercal\;w_\mr{e}^\intercal]^\intercal$ here, as the full vector $v$ is of size $60 \times 1\;(n+m+p = 60)$. The fixed point vector $x_\mr{e}$ is the equilibrium point of the state vector $[x_{1,j}(k)\;x_{2,j}(k)\;x_{3,j}(k)\;x_{4,j}(k)]^\intercal$ for all stations $j \in\{1,2,3,4\}$.
However, as mentioned before, there are more fixed points that satisfy the equality and inequality constraints. In this case, we get a plane of fixed points as the equality constraint matrix lacks the full rank by 2, i.e., rank($H_{\text{eq}}$) = 58 (the size of $H_\text{eq}$ is $60\times 60$, $n+m+p=60$). In order to find all the vectors that qualify as a fixed point of the system, we solve the following system of equality and inequality constraints for $\lambda=120$
\begin{align}
    H_{\text{eq}}\cdot v = h_{\text{eq}},\quad H_{\text{ineq}}\cdot v \leq h_{\text{ineq}}.
    \label{fixed pointfeasiblity2}
\end{align}
As the $H_\mr{eq}$ is of size $60\times 60$ and $H_\mr{ineq}$ is of size $6\times 60$ for the urban railway system, we do not present their value in this paper.
By using Gaussian elimination \citep{lay2003linear}, we find the following system of equations for $x\in\R^n$ from \eqref{fixed pointfeasiblity2}:
\begin{align}
   x \;&= x_\mathrm{p} +\sigma_1s_1+\sigma_2s_2\label{fixed pointset}\\
   x_\mr{p} &= \begin{array}{cccccccccccccccc}
\big[-180&
  -120&
  -840&
     0&
     0&
  -180&
  -360&
     0\hdots\\ \hdots
   -60&
  -120&
  -120&
     0&
     0&
     0&
     0&
     0\big]^\intercal
   \end{array}\notag\\
   s_1 &=\big[\begin{array}{cccccccccccccccc}
       1&
    1&
     0&
     0&
    1&
    1&
     0&
     0&
    1&
    1&
     0&
     0&
    1&
    1&
     0&
     0
   \end{array}\big]^\intercal\notag\\
   s_2 &=\big[\begin{array}{cccccccccccccccc}
    3&
    3&
   -8&
    0&
    3&
    1&
   -4&
    0&
    1&
    0&
   -2&
    0&
    0&
   -0.5&
   -1&
    0
   \end{array}\big]^\intercal\notag
\end{align}
where $\sigma_1,\sigma_2\in\R$ are the free variables.
\begin{remark}
    The direction vector $s_1$ is different from $s = [\one_{n_\mr{t}}^\intercal\;\zero_{n_\mr{q}}^\intercal]^\intercal$  as the model is represented as two temporal states and two quantity states per station $j$ rather than arranging all the temporal states and quantity states across the stations together.
\end{remark}
The fixed point set is the set of all $x$ of the form \eqref{fixed pointset} that satisfies the inequality 
$$H_{\text{ineq}}\cdot v\leq h_{\text{ineq}},$$
where $v = [x^\intercal\; y^\intercal\; w^\intercal]^\intercal$. Therefore, 
\begin{align}
    H_{\text{ineq}}\cdot v_\mathrm{p}+\sigma_1 H_{\text{ineq}}\cdot b_1+\sigma_2 H_{\text{ineq}}\cdot b_2 \leq h_{\text{ineq}}.
    \label{fixed pointspaceineq}
\end{align}

where the vector $v_\mr{p}=[x_\mr{p}^\intercal\; y_\mr{p}^\intercal\; w_\mr{p}^\intercal]^\intercal$ is the extended fixed point, and vectors $b_1$, $b_2$ are the extended basis vectors.
The term $H_{\text{ineq}}\cdot b_1$ is $\zero$. This also shows that the system is shift-invariant in the direction of $s_1$ and agrees with the time-invariance property. Hence, the free variable $\sigma_1$ is unbounded. However, the free variable $\sigma_2$ is bounded. The bounds on $\sigma_2$ can be found using inequality \eqref{fixed pointspaceineq} and they are obtained as 
\begin{align}
-45\leq \sigma_2\leq67.5.
\label{sigma2range}
\end{align}
The fixed points denote the steady-state values of the states $a_j,d_j,\rho_j,\sigma_j$ for the stations $j = 1,2,3,4$.  Also, when the system is initialized at a fixed point, the system stays at a stationary regime with a constant growth rate $\lambda$, i.e., $x(k) = x_\mr{e}+k\lambda$. For example, consider the growth rate and fixed point in \eqref{URSpara}. The steady-state arrival and departure time of train `$k$' at subsequent stations can be obtained as:
\begin{center}
    \begin{tabular}{|c|c|c|c|c|}
\hline
   \text{Stations}, $j$ &$1$&$2$&$3$&$4$ \\ \hline
    \text{Arrival time}, $a_j$ &$0$&180&360&540\\\hline
  \text{Departure time}, $d_j$&60&240&420&600\\\hline
\end{tabular}
\end{center}

Therefore, if the train `$k$' arrives and departs through the stations, 1,2,3,4 at the time instants as specified above, the next train (train `$k+1$') arrives and departs at stations $1,2,3,4$, at time  $a_j+120$, and $d_j+120$, respectively. Hence, we get a running timetable for the urban railway network from the fixed points and growth rate. The existence of a set of fixed points help us to choose the most suitable operating point for the application. For example, based on \eqref{fixed pointset} and \eqref{sigma2range}, the vector
\begin{align*}
    x_\mr{e_2} = \begin{array}{cccccccccccccccc}
        \big[0&   60&         0&         0&  180&  210&   60&         0\hdots\\\hdots  330&  375&   90&         0&  495&  547.5& 105&         0\big].
    \end{array}
\end{align*}
is also a fixed point for the system. However, in this vector, the difference between the departure time and arrival time at every station $j$ is not uniform as is the case for vector $x_\mr{e_1}$. Hence, the vector $x_\mr{e_2}$ does not provide a uniform timetable for the urban railway system. 
\begin{figure*}[t]
    \centering
     \begin{align}
    M&= \resizebox{0.6\textwidth}{!}{$\left[\begin{array}{cccccccccccccccc}
    1& 0& 0& 0& 0& 0& 0& 0& 0& 0& 0& 0& 0& 0& 0& 0\\
    1& 0& 0& 0& 0& 0& 0& 0& 0& 0& 0& 0& 0& 0& 0& 0\\
    0& 0& 1& 0& 0& 0& 0& 0& 0& 0& 0& 0& 0& 0& 0& 0\\
    0& 0& 0& 0& 0& 0& 0& 0& 0& 0& 0& 0& 0& 0& 0& 0\\
    1& 0& 0& 0& 0& 0& 0& 0& 0& 0& 0& 0& 0& 0& 0& 0\\
    1.33& 0& 0.33& 0& 0&-0.33& 0& 0.67& 0& 0& 0& 0& 0& 0& 0& 0\\
    0.67& 0& 0.67& 0& 0&-0.67& 0&1.33&0& 0& 0& 0& 0& 0& 0& 0\\
    0& 0& 0& 0& 0& 0& 0& 0& 0& 0& 0& 0& 0& 0& 0& 0\\
    1.33& 0& 0.33& 0& 0&-0.33& 0& 0.67& 0& 0& 0& 0& 0& 0& 0& 0\\
    2& 0& 0.67& 0& 0& -0.67& 0& 1.33& 0& -0.33& 0& 0.67& 0& 0& 0& 0\\
    1.33& 0& 0.67& 0& 0& -0.67& 0& 1.33& 0& -0.67& 0& 1.33& 0& 0& 0& 0\\
    0& 0& 0& 0& 0& 0& 0& 0& 0& 0& 0& 0& 0& 0& 0& 0\\
    2& 0& 0.67& 0& 0& -0.67& 0& 1.33& 0& -0.33& 0& 0.67& 0& 0& 0& 0\\
    3.11& 0& 1.11& 0& 0& -1.11& 0& 2.22& 0& -0.67& 0& 1.33& 0& -0.33& 0& 0.67\\
    2.22& 0& 0.89& 0& 0& -0.89& 0& 1.78& 0& -0.67& 0& 1.33& 0&-0.67& 0& 1.33\\
    0& 0& 0& 0& 0& 0& 0& 0& 0& 0& 0& 0& 0& 0& 0& 0
    \end{array}\right]$} \label{matrixM}\\
        H &= \resizebox{0.6\textwidth}{!}{$\left[\begin{array}{cccccccccccccccc}
        0.33& 0& 0.33& 0& 0& -0.33& 0& 0.67& 0& 0& 0& 0& 0& 0& 0& 0\\  
        0.67& 0& 0.33& 0& 0& -0.33& 0& 0.67& 0& -0.33& 0& 0.67& 0& 0& 0& 0\\  
        1.11& 0& 0.44& 0& 0& -0.44& 0& 0.89& 0& -0.33& 0& 0.67& 0& -0.33& 0& 0.67\\  
        -1& 0& 0& 0& 0& 1& 0& 0& 0& 0& 0& 0& 0& 0& 0& 0\\
        -1.33& 0& -0.33& 0& 0& 0.33& 0& -0.67& 0& 1& 0& 0& 0& 0& 0& 0\\  
         -2& 0& -0.67& 0& 0& 0.67& 0& -1.33& 0& 0.33& 0& -0.67& 0& 1& 0& 0 
        \end{array}\right]$},\; h = \begin{bmatrix}
   15\\
   15\\
  15\\
   30\\
   30\\
   30
        \end{bmatrix}
        \label{polyurban}
    \end{align}
\end{figure*}
Further, to assess the stability of the system operating at a fixed point, we first derive an equivalent linearized system around the fixed point $x_\mr{e_1}$. This is done by first getting a normalized system using \eqref{elaboration_norm} and then linearizing it using \eqref{linearsysimp}.  The linearized system is obtained as 
\begin{align}
    \tilde{x}(k) = M\cdot \tilde{x}(k-1), \;\Tx \in \Omega
    \label{linsysurban}
\end{align}
where $M$ is given in \eqref{matrixM}, the region $\Omega: H\cdot \Tx \leq h$, and $H$, $h$ are given in \eqref{polyurban}. 

The matrix $M$ has two multiplicative eigenvalues equal to $1$ with independent eigenvectors, and all the other multiplicative eigenvalues are less than one. Therefore, by using Proposition \ref{propstability}, we find that the system \eqref{linsysurban} is stable. Note that the rank deficiency of the matrix $H_\mr{eq}$ (=2) and the number of multiplicative eigenvalues of the matrix $M$ that are equal to $1$ are the same in accordance with the Corollary \ref{corrankeigrelation}.
\section{Conclusions}
\label{sec7}
In this paper, we have presented a set of linear programming problems to calculate the two most important parameters, the growth rates, and equilibrium points of an implicit MMPS system. Using these parameters, we have derived a normalized MMPS system with structural properties, which makes it easier to analyze them. Then, we have derived a linear system around the fixed points of the normalized system to study the local stability of an implicit autonomous MMPS system. These results will be useful when building a closed-loop controller for the system.

In the future, we will extend our work to investigate closed-loop control problems for discrete-event systems modeled as MMPS systems. Feedback control will be implemented using control Lyapunov functions.   
\bibliography{ifacconf}             

\end{document}